\documentclass[11pt]{article}
\usepackage{amsmath,amssymb,amsbsy,amsgen,amsopn,amsthm}
\usepackage{mathrsfs}
\usepackage{multirow,booktabs}
\usepackage{threeparttable}
\usepackage{hyperref}
\usepackage{times}
\usepackage[margin=1in]{geometry}

\newtheorem{theorem}{Theorem}
\newtheorem{lemma}{Lemma}
\newtheorem{corollary}{Corollary}

\newcommand{\Rset}{\mathbb{R}}

\newcommand{\Mfam}{\mathcal{M}}
\newcommand{\Afam}{\mathcal{A}}
\newcommand{\Cfam}{\mathcal{C}}

\newcommand{\Vfam}{\mathcal{V}}
\newcommand{\Lfam}{\mathcal{L}}
\newcommand{\Bfam}{\mathcal{B}}
\newcommand{\Pfam}{\mathcal{P}}

\newcommand{\LP}{{\tt PCLP}}
\newcommand{\NPCLP}{{\tt NPCLP}}
\newcommand{\CoreLP}{{\tt SimpleLP}}
\newcommand{\CoreDual}{{\tt SimpleDual}}
\newcommand{\true}{\mbox{true}}
\newcommand{\false}{\mbox{false}}

\usepackage[linesnumbered,boxed]{algorithm2e}
\usepackage{tikz}
\usetikzlibrary{backgrounds}
\usetikzlibrary{snakes}
\usetikzlibrary{shapes}
\usetikzlibrary{trees}  
\tikzstyle{vertex}=[circle,draw, inner sep=2.3pt]
\tikzstyle{terminal}=[circle,draw,fill, inner sep=2.3pt]
\tikzstyle{biset}=[fill=gray,fill opacity=.2, line width=8pt, draw=gray]
\tikzstyle{biset outer}=[dotted,line width=2pt]

\title{Spider covers for prize-collecting network activation problem}

\author{Takuro Fukunaga\footnote{National Institute of Informatics,
2-1-2 Hitotsubashi, Chiyoda-ku, Tokyo, Japan.
JST, ERATO, Kawarabayashi Large Graph Project, Japan.
Email: takuro@nii.ac.jp}}

\date{}

\begin{document}
\maketitle
\begin{abstract}
In network activation problem, each edge in a graph is associated with an
 activation function that decides 
 whether the edge is activated from weights assigned
 to its end nodes.
 The feasible solutions of the problem are node weights such that the
 activated edges form graphs of required connectivity,
 and the objective is to find a feasible solution minimizing its total
 weight.
 In this paper, we consider a prize-collecting version of the network
 activation problem and present the first nontrivial approximation algorithms.
 Our algorithms are based on a new linear programming relaxation of the problem.
 They round optimal solutions for the relaxation by repeatedly computing
 node weights
 activating subgraphs, called spiders, which are known to be useful for
 approximating the network activation problem.
 For the problem with node-connectivity requirements, we also present a new
 potential function on uncrossable biset families
 and use it to analyze our algorithms.
\end{abstract}

\section{Introduction}
\subsection{Problem}
{\em Network activation problem}
is a problem of 
activating a well-connected network by assigning weights to
nodes.
The problem is formally described as follows. Given a graph $G=(V,E)$ and
a set $W$ of non-negative real numbers such that $0 \in W$ and
$i + j \in W$ for any $i,j \in W$,
a solution in the problem is a
node weight function $w\colon V \rightarrow W$.
For $u,v\in V$,
let $\{u,v\}$ and $uv$ 
denote the unordered and ordered pairs of $u$ and $v$, respectively.
Each edge $\{u,v\} \in E$ is associated with an activation function $\psi^{uv}: W
\times W \rightarrow \{\true,\false\}$ such that
$\psi^{uv}(i,j)=\psi^{vu}(j,i)$ holds for any $i,j \in W$.
In this paper, each activation function $\psi^{uv}$ is supposed to be
{\em monotone}, i.e.,
if $\psi^{uv}(i,j)=\true$ for some $i,j \in W$,
then $\psi^{uv}(i',j')=\true$ for any $i',j'\in W$ with $i' \geq i$ and $j' \geq j$.
An edge $\{u,v\}$ is \emph{activated} by $w$ if
$\psi^{uv}(w(u),w(v))=\true$.
Let $E_{w}$ be the set of edges activated by $w$ in $E$.
A node weight function $w$ is feasible 
in the network activation problem
if $E_{w}$ satisfies given constraints, and the objective of the problem is to find a feasible
node weight function $w$ that minimizes $\sum_{v \in V}w(v)$, denoted by
$w(V)$. 
We assume throughout the paper that $G$ is undirected even though
the problem can be defined for directed graphs as well.

In this paper, we pose connectivity constraints 
on the set $E_w$ of activated edges.
Namely, we are given
demand pairs $\{s_1,t_1\},\ldots,\{s_d,t_d\}\subseteq V$ associated with
connectivity requirements $r_1,\ldots,r_d$ defined as natural numbers.
$[d]$ denotes $\{1,\ldots,d\}$, $k$ denotes $\max_{i \in [d]}r_i$,
and a node that participates in some demand pair is called a {\em terminal}.
The constraints require that
the connectivity between $s_i$ and $t_i$ in
the graph $(V,E_w)$ is at least $r_i$ for each $i \in [d]$.
We consider three definitions of the connectivity:
edge-connectivity, node-connectivity, and element-connectivity.
The edge-connectivity between two nodes $u$ and $v$ is the maximum
number of edge-disjoint paths between $u$ and $v$,
and the node-connectivity 
between $u$ and $v$ is the maximum
number of inner disjoint paths between $u$ and $v$.
The element-connectivity is defined only for pairs of terminals,
and for two terminals $u$ and $v$, it is defined as the maximum number
of paths between them that are disjoint in edges and in non-terminal nodes.
The edge-connectivity network activation problem denotes the problem with the edge-connectivity
constraints.
The node- and the element-connectivity network activation problems 
are defined similarly.

The network activation problem is closely related to
the \emph{survivable network design problem} (SNDP),
a problem of constructing a cheap network that is sufficiently connected.
A feasible solution to the SNDP is a subgraph $(V,F)$ of a given graph $G=(V,E)$
that satisfies the connectivity constraints.
There are two popular variations, called the edge- and
node-weighted SNDPs.
In the edge-weighted SNDP, each edge in the graph is associated with a
weight $w(e)$, and the objective is to minimize the weight $w(F)$ of $F$ defined as
$\sum_{e \in F}w(e)$.
In the node-weighted SNDP, 
a weight $w(v)$ is given for each node $v \in V$, and the objective is
to minimize $\sum_{v \in V(F)} w(v)$, where $V(F)$ denotes the set of
end nodes of edges in $F$.
We denote $\sum_{v \in V(F)} w(v)$ by $w(V(F))$ in the sequel.
It is known that the node-weighted SNDP generalizes the edge-weighted
SNDP.

It can be seen that the network activation problem extends the node-weighted SNDP.
Given node weights $w'\colon V\rightarrow \Rset_{\geq 0}$,
let $W=\{w'(v)\colon v \in V \} \cup \{0\}$,
and define a monotone activation function $\psi^{uv}$ for $\{u,v\}\in E$ so that
$\psi^{uv}(i,j)=\true$ if and only if $i\geq w'(u)$ and $j \geq w'(v)$.
A minimal solution $w\colon V\rightarrow W$ to the network activation
problem with these activation functions
does not assign a weight larger than $w'(v)$ to $v\in V$.
Hence, if an edge activated by $w$ is incident to a node $v$, 
then $w(v)=w'(v)$ holds without loss of generality.
Therefore, the node-weighted SNDP with $w'$ is equivalent to the network activation problem
with $\psi$ defined from $w'$.

The extension from the SNDP to the network activation problem is not
only important from a technical viewpoint but also for practical
reasons. In the node-weighted SNDP, for each node, one is required to decide whether it
is chosen.  In contrast, the network activation problem demands
a decision concerning which weight is assigned to a node. In other words, the
network activation problem admits more than two choices while the
node-weighted SNDP admits only two choices for each node.  This rich
structure of the network activation problem enables to capture many
problems motivated by realistic applications.  In fact,
Panigrahi~\cite{Panigrahi11wireless} discussed numerous applications to
wireless networks.  In wireless networks, the success of communication
between two base stations depends on factors such as physical
obstacles between them, positions of antennas, and signal strength.
Panigrahi suggested that many problems related to wireless networks can
be modeled by the network activation problem.

Our main contribution in this paper is to develop algorithms for 
a prize-collecting version of the network activation
problem, which we call the {\em prize-collecting network activation
problem} (PCNAP). 
In the PCNAP,
each demand pair $\{s_i,t_i\}$ is associated with not only a
connectivity requirement $r_i$, but also a non-negative real number
$\pi_i$, which is called the \emph{penalty}.
The edge set $E_w$ activated by a solution $w$ 
is allowed to violate the
connectivity requirements,
but it has to pay the penalty $\pi_i$ if it does not satisfy the connectivity
requirement for $\{s_i,t_i\}$. 
The objective of the PCNAP
is to minimize the sum of 
$w(V)$ and the penalties we have to pay.

We also consider two variations of the PCNAP.
The {\em rooted node-connectivity PCNAP} is a special case of the
node-connectivity PCNAP such that 
a root node $s \in V$ is specified and the demand pairs
are $\{s,t_1\},\ldots,\{s,t_d\}$.
In the {\em subset node-connectivity PCNAP},
terminals $t_1,\ldots,t_d \in V$ 
and penalties $\pi_1,\ldots,\pi_d$ are given instead of demand pairs.
Let $E_w$ be the set of activated edges. 
In addition to node-weights,
a solution chooses $U \subseteq [d]$ such that 
every pair of terminals $t_i$ and $t_j$ 
with $i,j \in U$
is $k$-connected in the graph $(V,E_w)$.
The penalty is $\sum_{i \in [d]\setminus U}\pi_i$.
We note that the subset node-connectivity PCNAP is not a special case of
the node-connectivity PCNAP because the above setting cannot be represented by 
connectivity demands and penalties on terminal pairs.

In all of the known applications, it is reasonable to assume 
$|W|={\rm poly}(|V|)$.
In fact, all previous research~\cite{Nutov13activation,Panigrahi11wireless} studied the network activation problem 
under this assumption.
In this paper, we proceed on the same assumption and 
design algorithms
that run in polynomial time of $|W|$ and the size of $G$.

\subsection{Related work} 

The SNDP is a well-studied optimization
problem, and there are substantial number of studies regarding algorithms for
it. The best known approximation factors for the edge-weighted SNDP are two
for the edge-~\cite{Jain01} and element-connectivity~\cite{FleischerJW06},
and $O(k^3 \log |V|)$ for node-connectivity~\cite{ChuzhoyK12}. 
For the node-weighted SNDP, 
Nutov~\cite {Nutov10node-weights} gave an $O(k\log |V|)$-approximation algorithm with edge-connectivity requirements, and
element-connectivity requirements in~\cite{Nutov12uncrossable}. His algorithm
is based on an algorithm for the problem of covering uncrossable biset
families by edges, where a biset is an ordered pair of two node sets, and an uncrossable family 
is a family closed under some uncrossing operations (we will present their formal definitions later). 
However, his analysis of the
algorithm for covering uncrossable biset families has an error.
We will explain it in Section~\ref{sec.potential}.

The prize-collecting SNDP has also been well studied.
As for edge-weighted graphs,
we refer to only Hajiaghayi et al.~\cite{HajiaghayiKKN12} 
whereas many papers studied related problems such as the prize-collecting Steiner tree and forest.
Recently much attention has been paid to node-weighted graphs.
K\"onemann, Sadeghian, and Sanit\`a~\cite{Konemann13CoRR} gave an $O(\log
|V|)$-approximation algorithm for the prize-collecting node-weighted Steiner tree problem.
Their algorithm has the Lagrangian multiplier preserving property, which is
useful in many contexts.
They also pointed out a technical error in Moss and Rabani~\cite{MossR07}.
Bateni, Hajiaghayi, and Liaghat~\cite{BateniHL13} gave an $O(\log
|V|)$-approximation algorithm for the prize-collecting node-weighted Steiner forest
problem with application to the budgeted Steiner tree problem.
Chekuri, Ene, and Vakilian~\cite{ChekuriEV12} gave an $O(k^2 \log
|V|)$-approximation for the prize-collecting SNDP with edge-connectivity
requirements, which they later improved to $O(k \log |V|)$-approximation
and also extended to 
the element-connectivity
requirements (refer to \cite{Vakilian13}).
We note that the proof in~\cite{Vakilian13} implies that the 
algorithm in~\cite{Nutov12uncrossable} works for the node-weighted SNDP
with element-connectivity requirements,
as Nutov originally claimed, even though his analysis of the algorithm for covering
uncrossable biset families is not correct in general.
We also note that 
the algorithm for the element-connectivity requirements in~\cite{Vakilian13} implies 
$O(k^4\log |V|)$-approximation for node-connectivity requirements,
using the reduction from node-connectivity requirements to 
the element-connectivity requirements presented by Chuzhoy and Khanna~\cite{ChuzhoyK12}.

Concerning the network activation problem, 
Panigrahi~\cite{Panigrahi11wireless} gave $O(\log |V|)$-approximation
algorithms for $k\leq 2$ and proved
that it is NP-hard to obtain an $o(\log |V|)$-approximation algorithm even
when activated edges are required to be a spanning tree.
Nutov~\cite{Nutov13activation} presented approximation algorithms for
higher connectivity requirements, including $O(k\log |V|)$-approximation for
the edge- and element-connectivity and $O(k^4 \log^2 |V|)$-approximation for
the node-connectivity.
He also discussed special node-connectivity requirements such as rooted and
subset requirements.
These results are built based on his research in~\cite{Nutov12uncrossable} for
covering uncrossable biset families. 
This contains an error as
mentioned above, and 
the rectification offered in~\cite{Vakilian13} cannot be extended to the network
activation problem.
Therefore, the network activation problem currently
has no non-trivial algorithms for the element- and node-connectivity.
One contribution of this paper is to rectify the Nutov's error and to provide
algorithms for these problems.

An important factor
in most of the research mentioned above is the {\em greedy spider cover algorithm}.
The notion of {\em spiders} was invented by Klein and Ravi~\cite{KleinR95} in
order to solve
the node-weighted Steiner tree problem. 
It was originally defined as a tree that admits at most one
node of degree larger than two and that spans at least two terminals.
The node of degree larger than two is called the {\em head}, 
and nodes of degree one are called the {\em feet}
of the spider. 
It is supposed without loss of generality that each foot of a spider is a terminal.
If all nodes have degrees of at most two, then an arbitrary node is chosen to be the head.
Klein and Ravi~\cite{KleinR95} proved that any Steiner tree can be
decomposed into node-disjoint spiders so that each terminal is included by
some spider. The {\em density} of a subgraph is defined as 
its node weight divided by the number of terminals included by it.
The decomposition theorem implies that there exists a spider with a
density of at
most that of Steiner trees. 
Since contracting a spider with $f$ feet decreases the number of
terminals by at least $f-1$,
a greedy algorithm to 
repeatedly contract minimum density spiders
achieves $O(\log |V|)$-approximation.
Minimum density spiders are hard to compute but their relaxations can be
computed by a simple algorithm that involves first guessing the place of the head and 
number of feet, which is possible
because there are only $|V|$ options for each. Let $h$ be the head, and
$f$ be the number of feet.
We then compute a shortest path from $h$ to each terminal, and
choose the $f$ shortest paths from them.
The union of these shortest paths is not necessarily a spider, but its
density is at most that of spiders, and contracting the union can play the
same role as contracting spiders.
Nutov~\cite{Nutov10node-weights,Nutov12uncrossable,Nutov13activation}
extended the notion of spiders to uncrossable biset families, and demonstrated in the
sequence of his research that they are useful for the node-weighted SNDP and the
network activation problem.

\subsection{Our results}

The main result in this paper is to present approximation algorithms for
the PCNAP.
Our algorithms achieve $O(k\log |V|)$-approximation for the
edge-connectivity PCNAP, and
$O(k^2 \log |V|)$-approximation for the element-connectivity PCNAP.
Table~\ref{table.results} summarizes the approximation factors achieved by our algorithms
and previous studies.
Using decompositions of connectivity requirements given in 
\cite{ChuzhoyK12,Nutov12uncrossable,Nutov12subset},
we can also achieve approximation factors
$O(k^5 \log^2 |V|)$ for the node-connectivity PCNAP
and $O(k^3 \log |V|)$ for the rooted and
subset node-connectivity PCNAPs.
Our results 
give the first non-trivial algorithms for the PCNAP.
We also recall that, 
besides our algorithms,
no algorithms are known
even for the element- and node-connectivity network activation problems
because the analysis of the algorithms claimed 
by Nutov~\cite{Nutov12uncrossable,Nutov13activation}
contains an error.
For wireless networks, it is natural to consider node-connectivity, 
which represents tolerance against node failures,
rather than edge-connectivity, which represents tolerance against link
failures.
Hence, our results are important for not only theory but also applications.

\begin{table*}
\centering
\caption{Approximation factors for the edge-weighted SNDP, node-weighted SNDP, and 
the network activation problem}
\label{table.results}

\begin{threeparttable}
\begin{tabular}{lllll}
\toprule
 & \multicolumn{2}{l}{non-prize-collecting} & \multicolumn{2}{l}{prize-collecting}\\
\midrule
{\bf edge-connectivity}&&&&\\
edge-weighted SNDP & 2 & Jain~\cite{Jain01} & 2.54 & Hajiaghayi et al.~\cite{HajiaghayiKKN12} \\
node-weighted SNDP & 
	$O(k\log |V|)$ & Nutov~\cite{Nutov10node-weights} 
	& $O(k\log |V|)$ & Chekuri et al.~\cite{ChekuriEV12}\\
network activation & 
	$O(k\log |V|)$ & Nutov~\cite{Nutov13activation}
	& $O(k \log |V|)$ & [this paper]\\
\midrule
{\bf element-connectivity}&&&&\\
edge-weighted SNDP
	& 2 & Fleischer et al.~\cite{FleischerJW06}
	& 2.54 &Hajiaghayi et al.~\cite{HajiaghayiKKN12}\\
node-weighted SNDP
	& $O(k \log |V|)$ & Vakilian \cite{Vakilian13}\tnote{1}
	& $O(k\log |V|)$ & Vakilian \cite{Vakilian13}\\
network activation  
 & $O(k^2 \log |V|)$ & [this paper]\tnote{1}
 & $O(k^2 \log |V|)$ & [this paper]\\
\bottomrule
\end{tabular}
\begin{tablenotes}\footnotesize
\item[1] Nutov~\cite{Nutov12uncrossable,Nutov13activation} claimed $O(k \log |V|)$-approximation
algorithms for the node-weighted SNDP and the network activation problem with element-connectivity 
constraints, but these contained an error.
\end{tablenotes}
\end{threeparttable}
\end{table*}


Let us present a high level overview of our algorithms.
Our algorithms first reduce the problem with high connectivity requirements
to the \emph{augmentation problem},
which asks to increase the connectivity of demand pairs by one.
This is a standard trick for SNDP,
and we will show in Section~\ref{sec.preliminaries} that this trick can work
even for the PCNAP.
Then, our algorithms compute an optimal solution to an LP relaxation, and 
discards some of the demand pairs according to the optimal solution, 
which is a popular way to deal with prize-collecting problems
since Bienstock~et~al.~\cite{BienstockGSW93}.
In the last step, the algorithms solves the problem using the greedy spider cover algorithm.
To obtain an approximation guarantee, we are required to show that 
the minimum density of spiders can be bounded
in terms of the optimal value of the LP relaxation.
We achieve this by presenting a primal-dual algorithm for computing spiders,
which is the same approach as~\cite{ChekuriEV12,BateniHL13,Vakilian13}.

As observed from this overview,
our algorithms rely on many ideas
given in the previous studies on the prize-collecting
SNDP and the network activation problem.
However, it is highly nontrivial to apply these ideas for the PCNAP, and 
we required several new ideas to obtain our algorithms.
Specifically,
the technical contributions of the present paper are the following three new findings:
an LP relaxation of the problem,
a primal-dual algorithm for computing spiders, and a 
potential
function for analyzing the greedy spider cover algorithm.
Below we explain these one by one.

\subsubsection*{LP relaxation}
Nutov's spider decomposition theorem is useful for 
the biset covering problem defined from the
SNDP and the network activation problem, but 
we have to strengthen it for solving their prize-collecting versions. 
We define an LP relaxation of the problem
and compare the minimum density of spiders with the density of
fractional solutions feasible to this relaxation. The same attempt has
been made previously by
\cite{BateniHL13,ChekuriEV12,Konemann13CoRR} for the node-weighted SNDP,
but our situation is much more complicated.
Each connectivity requirement in the node-weighted SNDP can be simply represented 
by demands on the number of chosen nodes in node cuts of graphs,
which naturally formulates an LP relaxation that performs well.
On the other hand,
the network
activation problem requires to decide which edges are
activated for covering bisets in addition to the decision on which weights are assigned to nodes
for activating the edges.
Hence an LP relaxation for the network activation problem needs
variables corresponding to edges and nodes whereas that for the node-weighted SNDP 
needs only variables corresponding to nodes.
However, 
dealing with both edge and node variables introduces 
a large integrality gap into
a natural LP relaxation for the network
activation problem, as we will see in Section~\ref{sec.LP}.
Hence we require to formulate an LP relaxation carefully.

In the present paper, 
we propose a new LP that lifts the natural LP relaxation for the PCNAP.
It is non-trivial even to see that our LP relaxes the PCNAP.
We prove it
using the structure of uncrossable biset families, wherein
any uncrossable biset family
can be decomposed into 
a polynomial number of ring biset families,
and 
the degree of each node is at most two in
any minimal edge cover of a ring biset family.
In addition,
the main result in this paper implies that
our LP has small integrality gap.

Let us mention that
the idea on formulating our LP relaxation is potentially useful for other covering problems.
The author pointed out
in his recent work~\cite{Fukunaga14}
that
a natural LP relaxation has a large integrality gap for many covering problems
in node-weighted graphs. He also presented several tight approximation algorithms
using the LP relaxations designed based on the idea we propose in the present paper.

\subsubsection*{Primal-dual algorithm for computing spiders}
For bounding the minimum density of spiders in terms of optimal
values of our relaxation, we will present a primal-dual algorithm for computing
spiders. 
Usually, a primal-dual algorithm computes fractional solutions feasible to
the dual of an LP relaxation together with primal solutions, but this seems
difficult for our relaxation because of its complicated form.
Hence, our algorithm does not directly
compute solutions feasible to the dual of our relaxation.
Instead, we define another LP simpler than our relaxation,
and our algorithm computes feasible solutions to the dual of this
simpler LP.
Although the simpler LP does not relax our relaxation, 
we can show that it is within a constant factor of our relaxation
if biset families are restricted to laminar families of cores, which are
bisets that do not include more than one minimal biset.
Our primal-dual algorithm 
computes dual solutions
that assign non-zero values only to variables corresponding to cores in 
laminar families.
Hence, the density of spiders can be analyzed in terms of our relaxation.

Summarizing, our algorithm uses two different LPs:
the LP based on the structure of uncrossable biset families
is used for deciding which demand pairs are discarded in the first step, and
the simpler LP 
with laminar core families is used
in the second step that iterates choosing spiders.
We note that the simpler LP cannot be used 
in the first step because of two reasons.
First,
we do not know beforehand which laminar core families will be used, and second, we have different laminar families in distinct iterations.

Although
our primal-dual algorithm for the simpler LP 
seems to be similar to primal-dual algorithms known for related
problems, 
its design and analysis is not trivial.
 One reason for this is the
existence of more than one choices of weights for each end node of activated edges as we have already
mentioned. Another reason is the involved structure of bisets. Since a biset is defined as an ordered pair of
two node sets, covering a biset family by edges is much more difficult problem than covering a set family,
for which primal-dual algorithms are often studied. Indeed, our algorithm utilizes many non-trivial
properties of uncrossable biset families.

\subsubsection*{Potential function for analyzing greedy spider cover algorithm}
Nutov~\cite{Nutov12uncrossable} claimed that
repeatedly choosing a constant approximation of minimum density spiders
achieves 
$O(\log |V|)$-approximation for covering uncrossable biset families.
This claim is true if biset families are defined from
edge-connectivity requirements.
However it is not true for all uncrossable biset families.
The claim is based on the fact that contracting a spider with $f$ feet
decreases the number of minimal bisets by a constant fraction of $f$.
However there is a case in which contracting a spider does not decrease
the number at all (see Section~\ref{sec.potential}).
Chekuri, Ene, and Vakilian~\cite{Vakilian13}
showed that the claim is true for biset families arising from the
node-weighted SNDP,
but it cannot be extended to arbitrary uncrossable biset families,
including those from the network activation problem.

To rectify this situation,
we will define a new potential function.
The new potential function depends on
the numbers of minimal bisets and
 nodes shared by more
than two minimal bisets.
If the number of minimal bisets does not decrease considerably when a
spider is selected, 
many new minimal bisets share the head of the spider.
This fact motivates the definition of the potential function.

With this new potential function,
the definition of density of an edge set will be changed to the total weight for
activating it divided by the value of the potential function.
We cannot prove that the minimum density of spiders is at most that of
biset family covers after changing the definition of density.
Instead, we will show that a spider minimizing the density in the old definition
approximates 
the density of biset family covers in the new definition within a factor of $O(k)$.
This proves that the greedy spider covering algorithm achieves
 $O(k\log |V|)$-approximation for the biset
covering problem with uncrossable biset families.
Since Klein and Ravi~\cite{KleinR95},
the greedy spider cover algorithms have been applied to many problems
related to the node-weighted SNDP.
Considering this usefulness of the greedy spider cover algorithms,
our potential function is of independent interest
because it is required for analyzing the algorithms for
uncrossable biset families.

\subsection{Roadmap}
The remainder of this paper is organized as follows. 
Section~\ref{sec.preliminaries}
presents reduction from the PCNAP to the
augmentation problem
and introduces preliminary facts on biset families. 
Section~\ref{sec.LP}
defines our LP relaxation.
Section~\ref{sec.primal-dual} presents our primal-dual algorithm for
computing spiders, and Section~\ref{sec.potential}
presents a new potential function for analyzing the greedy spider covers.
Section~\ref{sec.algorithm} presents our algorithms, with
Section~\ref{sec.conclusion}
concluding this paper.

\section{Preliminaries}\label{sec.preliminaries}

\subsection{Reduction to the augmentation problem}\label{sec.augmentation}
First, we define the augmentation problem in detail.
We assume that there are two edge sets $E_0$ and $E$, and
activation functions are given for edges in $E$.
The connectivity of each 
demand pair $\{s_i,t_i\}$ is at least $k'-1$
in the graph $(V,E_0)$, 
and a subset $F$ of $E$ is feasible if
the connectivity of each demand pair in
$(V,E_0 \cup F)$ is at least $k'$.
The objective of the problem is to find a node weight function $w\colon V
\rightarrow W$ so that $E_w$ is feasible and
$w(V)$ is minimized.
In the prize-collecting augmentation problem, each demand pair
$\{s_i,t_i\}$
has a penalty $\pi_i$, and if the connectivity of $\{s_i,t_i\}$ is not
increased by $E_w$, then we
must pay the penalty. The objective of the prize-collecting
augmentation problem is to find a node weight function $w$ that
minimizes the sum of $w(V)$ and penalties of demand pairs of
connectivity smaller than $k'$ in $(V,E_0 \cup E_w)$.
PCNAP can be reduced to the
prize-collecting augmentation problem as follows.

\begin{theorem}\label{thm.augmentation}
 If the prize-collecting augmentation problem admits an
 $\alpha$-approximation algorithm, then PCNAP admits an $\alpha
 k$-approximation algorithm.
\end{theorem}
\begin{proof}
 We sequentially define instances of the prize-collecting augmentation
 problem. In the first instance,
 $E_0$ is set to be empty and $E$ is the edge set of the graph in the
 instance of the PCNAP. 
 Activation functions, demand pairs and their penalties are same as
 those in the PCNAP instance.
 The connectivity of each demand pair is $0$ in $(V,E_0)$,
 and the requirement of a demand pair is satisfied 
 if its connectivity is increased to at least one in $(V,E_0 \cup E_w)$.
 
 We define the $k'$-th instance after solving the $(k'-1)$-th instance.
 Let $w_{k'-1}$ be the node weights computed by the
 $\alpha$-approximation algorithm for the $(k'-1)$-th instance,
 and $D_{k'-1}$ be the set of indices of demand pairs that are
 satisfied by $w_{k'-1}$ in the $(k'-1)$-th instance.
 We move the edges activated by $w_{k'-1}$ from $E$ to $E_0$.
 For each $i \in D_{k'-1}$, the connectivity of $\{s_i,t_i\}$  is at least
 $k'-1$ in $(V,E_0)$ after the update.
 Let $I_{k'}=\{i \in D_{k'-1}\colon r_{i}\geq k'\}$.
 We define the demand pairs in the $k'$-th instance as $\{s_i,t_i\}$,
 $i\in I_{k'}$.
 The activation functions in the $k'$-th instance 
 are same as those in the PCNAP instance.

 We repeat the above sequence until the $k$-th instance is solved.
 Our solution to the PCNAP instance is $w=\sum_{k'=1}^k w_{k'}$.
 We prove that $w$ achieves $\alpha k$-approximation.
 Let $w^*$ be an optimal solution for the PCNAP instance, and $D^*=\{i
 \in [d] \colon \mbox{$\{s_i,t_i\}$ is satisfied by $E_{w^*}$}\}$.
 Then,
 the optimal value of the PCNAP instance is 
 $w^*(V) + \sum_{i \in [d]\setminus D^*}\pi_i$.
 If an edge is activated by $w^*$ in the PCNAP instance,
 then it is either in $E_0$ or is activated by $w^*$ in the $k'$-th
 instance of the prize-collecting augmentation problem.
 Hence, a demand pair $\{s_i,t_i\}$ with $i \in I_{k'}$
 is satisfied by $w^*$ if it is satisfied by $w^*$ in the PCNAP instance,
 implying that
 the objective value of $w^*$ 
 in the $k'$-th instance
 is at most $w^*(V) + \sum_{i \in I_{k'}\setminus D^*}\pi_i$.
 By the $\alpha$-approximability of $w_{k'}$,
 we have
\[
 w_{k'}(V) + \sum_{i \in I_{k'}\setminus D_{k'}}\pi_i \leq \alpha \left(w^*(V) + \sum_{i \in I_{k'}\setminus D^*}\pi_i\right).
\]
The objective value of $w$ in the PCNAP instance is 
\[
 \sum_{k'=1}^k \left(w_{k'}(V) + \sum_{i \in I_{k'}\setminus D_{k'}}\pi_i\right)
\leq 
 \alpha \sum_{k'=1}^k\left(w^*(V) + \sum_{i \in I_{k'}\setminus D^*}\pi_i\right)
\leq 
 \alpha k \left(w^*(V) + \sum_{i \in [d]\setminus D^*}\pi_i\right).
\]
\end{proof}

\subsection{Biset covering problem}\label{sec.biset}
Here, we formulate the prize-collecting augmentation problem 
as a problem of activating edges covering bisets.
A {\em biset} is an ordered pair $\hat{X}=(X,X^+)$ of subsets of $V$ such that $X \subseteq
X^+$. The former element of a biset is called the {\em inner-part} and the
letter is called the {\em outer-part}. 
We always let $X$ denote the inner-part of
a biset $\hat{X}$ and $X^+$ denote the outer-part of $\hat{X}$.
$X^+\setminus X$ is called the {\em boundary} of a biset $\hat{X}$ and is denoted
by $\Gamma(\hat{X})$.
For an edge set $E$, $\delta_E(\hat{X})$ denotes the set of edges in
$E$ that have one end-node in $X$ and the other in $V \setminus X^+$.
We say that an edge $e$ {\em covers} $\hat{X}$ if 
$e \in \delta_E(\hat{X})$, and a set $F$ of edges {\em covers} 
a biset family $\Vfam$ if each $\hat{X} \in \Vfam$ is covered by some
edge in $F$.

Let $i \in [d]$.
We say that a biset $\hat{X}$ {\em separates} a demand pair $\{s_i,t_i\}$ if 
$|X \cap \{s_i,t_i\}|=|\{s_i,t_i\}\setminus X^+|=1$.
We define $\Vfam^{\rm edge}_i$ as the family of bisets $\hat{X}$ such that
$X=X^+\subset V$, $|\delta_{E_0}(\hat{X})|=k-1$, and $\hat{X}$ separates 
the demand pair $\{s_i,t_i\}$.
According to Menger's theorem, $F \subseteq E$ 
increases the edge-connectivity of $\{s_i,t_i\}$ in the augmentation problem
if and only if 
$F$ covers $\Vfam^{\rm edge}_i$.
We define $\Vfam^{\rm node}_i$ as the family of bisets $\hat{X}$ such that
$|\delta_{E_0}(\hat{X})|+|\Gamma(\hat{X})|=k-1$ and $\hat{X}$ separates the demand pair $\{s_i,t_i\}$.
$F \subseteq E$ 
increases the node-connectivity of $\{s_i,t_i\}$ 
if and only if 
$F$ covers $\Vfam^{\rm node}_i$.
We define $\Vfam^{\rm ele}_i$ as the family of bisets $\hat{X} \in
\Vfam^{\rm node}_i$ such that
$\Gamma(\hat{X})\cap \{s_{i'},t_{i'}\}=\emptyset$ for each $i' \in [d]$.
$F \subseteq E$ 
increases the element-connectivity of $\{s_i,t_i\}$ 
if and only if 
$F$ covers $\Vfam^{\rm ele}_i$.

For two bisets $\hat{X}$ and $\hat{Y}$, we define
$\hat{X}\cap \hat{Y}= (X\cap Y, X^+ \cap Y^+)$, 
$\hat{X}\cup \hat{Y}= (X\cup Y, X^+ \cup Y^+)$, 
and 
$\hat{X}\setminus \hat{Y}= (X\setminus Y^+, X^+ \setminus Y)$.
A biset family $\Vfam$ is called {\em uncrossable}
if, for any $\hat{X},\hat{Y} \in \Vfam$, 
(i)~$\hat{X}\cap \hat{Y}, \hat{X}\cup \hat{Y}\in \Vfam$, or 
(ii)~$\hat{X}\setminus \hat{Y}, \hat{Y}\setminus \hat{X}\in \Vfam$ holds.
The following lemma indicates that the uncrossable biset families characterize
the augmentation problem with edge- and element-connectivity requirements.

\begin{lemma}\label{lem.edge-element}
For any $D\subseteq [d]$,
biset families
$\bigcup_{i \in D}\Vfam^{\rm edge}_i$ and
$\bigcup_{i \in D}\Vfam^{\rm ele}_i$
are uncrossable.
\end{lemma}

Lemma~\ref{lem.edge-element} follows from the submodularity and
posimodularity of $|\delta_{E_0}(\cdot)|$ and $|\Gamma(\cdot)|$, and a simple case analysis.
The same claim can be found in
\cite{FleischerJW06,Nutov12uncrossable}, and we recommend referring to
them for the proof of Lemma~\ref{lem.edge-element}.

By Lemma~\ref{lem.edge-element},
the problem of finding a minimum weight edge set covering 
a given uncrossable biset family
contains the augmentation problem with edge- or element-connectivity
requirements. The biset family
$\bigcup_{i \in D}\Vfam^{\rm node}_i$
defined from the node-connectivity requirements is not
necessarily uncrossable. However, it was shown previously in~\cite{ChuzhoyK12,Nutov12uncrossable,Nutov12subset}
that this family can be decomposed into uncrossable families, and the union of
covers of these uncrossable families gives a good approximate solution
for the node-connectivity augmentation problem.
We apply this approach for dealing with node-connectivity constraints
(see Section~\ref{sec.algorithm}).

We define the {\em biset covering problem}
as 
the problem of 
minimizing the sum of node weights under the constraint that the edges
activated by the node weights cover given biset families.
The prize-collecting version of the biset covering problem
is defined as follows.
Given an undirected graph $G=(V,E)$ such that each edge in $E$ is
associated with an activation function, 
demand pairs
$\{s_1,t_1\},\ldots,\{s_d,t_d\}$ with penalties
$\pi_1,\ldots,\pi_d$,
and a biset family $\Vfam$ on $V$.
For $i \in [d]$, let $\Vfam_i$ be the family of bisets in $\Vfam$
that separate $\{s_i,t_i\}$.
We say that $\hat{X} \in \Vfam$ is {\em violated} by an edge set $F\subseteq E$
if $\delta_F(\hat{X})=\emptyset$.
The penalty of $w\colon V \rightarrow W$ is $\sum \pi_i$ where the summation is
taken over all $i \in [d]$ such that $E_w$ violates some biset in $\Vfam_i$.
The objective of the problem is to find $w\colon V \rightarrow W$ that minimizes 
the sum of $w(V)$ and penalty of $w$.
This problem generalizes
the prize-collecting augmentation problem,
and hence, it suffices to present an algorithm for this problem.

Our results require several properties of uncrossable biset families.
We say that bisets $\hat{X}$ and $\hat{Y}$ are {\em strongly disjoint}
when both $X\cap Y^+=\emptyset$ and $X^+ \cap Y=\emptyset$ hold.
When $X \subseteq Y$ and $X^+ \subseteq Y^+$,
we say $\hat{X}\subseteq \hat{Y}$.
Minimality and maximality in a biset family are defined with regard to
inclusion.
A biset family $\Vfam$ is called {\em strongly laminar}
when, if $\hat{X},\hat{Y}\in \Vfam$ are not strongly disjoint, then 
they are comparable (i.e.,
$\hat{X}\subseteq \hat{Y}$ or $\hat{Y}\subseteq \hat{X}$).
A minimal biset in a biset family $\Vfam$ is called a {\em min-core}, and
$\Mfam_{\Vfam}$ denotes the family of min-cores in $\Vfam$.
A biset is called a {\em core} if it includes only one min-core, and 
$\Cfam_{\Vfam}$ denotes the family of cores in $\Vfam$,
where min-cores are also cores.
When $\Vfam$ is clear from the context, we may simply denote them
by $\Mfam$ and $\Cfam$.

For a biset family $\Vfam$, biset $\hat{X}$, and node $v$, 
$\Vfam(\hat{X})$ denotes $\{\hat{Y} \in \Vfam \colon
\hat{X}\subseteq \hat{Y}\}$
and $\Vfam(\hat{X},v)$ denotes $\{\hat{Y} \in \Vfam(\hat{X}) \colon
v \not\in Y^+\}$.
A biset family $\Vfam$ is called a {\em ring-family}
if
$\hat{X}\cap \hat{Y}, \hat{X}\cup \hat{Y}\in \Vfam$ hold
for any $\hat{X},\hat{Y} \in \Vfam$.
A maximal biset in a ring-family is unique because ring-families are closed under union.

\begin{lemma}\label{lem.uncrossable}
If $\Vfam$ is an uncrossable family of bisets, then the following properties hold\/{\rm :}
\begin{enumerate}
 \item[\rm (i)] $\Cfam(\hat{X})$ is a ring-family for any 
	      $\hat{X} \in \Mfam$.
 \item[\rm (ii)] Let $\hat{X},\hat{Y} \in \Mfam$ be distinct min-cores.
	      For any $\hat{X}' \in \Cfam(\hat{X})$ and $\hat{Y}' \in
	      \Cfam(\hat{Y})$, 
	      both $\hat{X}'\setminus \hat{Y}'\in \Cfam(\hat{X})$
	      and $\hat{Y}'\setminus \hat{X}'\in \Cfam(\hat{Y})$
	      hold.
 \item[\rm (iii)] Let $\hat{X},\hat{Y} \in \Mfam$ be distinct min-cores.
	      Then $\hat{Y}$ is strongly disjoint with
	      any $\hat{X}' \in \Cfam(\hat{X})$.
	      In particular, min-cores are pairwise strongly disjoint.
\end{enumerate}
\end{lemma}

The proof of Lemma~\ref{lem.uncrossable} can be found
in~\cite{Nutov12uncrossable}. 

For a biset family $\Vfam$ and an edge set $F$, let $\Vfam_F=\{\hat{X}
\in \Vfam\colon \delta_F(\hat{X})=\emptyset\}$.
The following lemma is required when we compute solutions recursively.

\begin{lemma}\label{lem.residual-uncrossable}
 Let $\Vfam$ be a family of bisets and $F\subseteq E$.
 Then $\Vfam_F$ is
 uncrossable if $\Vfam$ is uncrossable.
 $\Vfam_F$ is a ring-family if $\Vfam$ is a ring-family.
\end{lemma}
\begin{proof}
If bisets $\hat{X}$ and $\hat{Y}$ satisfy
$\delta_F(\hat{X})=\delta_F(\hat{Y})=\emptyset$,
then 
all $\delta_F(\hat{X}\cap \hat{Y})$, $\delta_F(\hat{X}\cup \hat{Y})$,
$\delta_F(\hat{X}\setminus \hat{Y})$, and $\delta_F(\hat{Y}\setminus
 \hat{X})$
are empty.
The claim follows from this fact.
\end{proof}

Below, we consider directed edges for technical reasons.
$A$ denotes the set of directed edges obtained by orienting the edges
in $E$ in both directions.
$\delta^-_A(\hat{X})$ denotes $\{uv \in A\colon v \in X, u \in V
\setminus X^+\}$ for a biset $\hat{X}$.
We say that a directed edge $e$ covers a biset $\hat{X}$ if $e \in
\delta^-_A(\hat{X})$,
and a set $F$ of directed edges covers a biset family $\Vfam$ if each
biset in $\Vfam$ is covered by some edge in $F$.
The following lemma will be required to prove that our LP relaxes the
prize-collecting biset covering problem.

\begin{lemma}\label{lem.ring-cover}
 Let $F$ be an inclusion-wise minimal set of directed edges 
 that covers a ring-family $\Vfam$ of bisets.
 Then the in-degree and out-degree of each node in the graph $(V,F)$ 
 is at most one.
\end{lemma}
\begin{proof}
 Let $v \in V$.
 We see that at most one edge in $F$ leaves $v$.
 For arriving at a contradiction, suppose that
 $F$ contains two edges $e=vu$ and $e'=vu'$.
 By the minimality of $F$,
 there exist $\hat{X} \in \Vfam$  
 with $\delta^-_{F}(\hat{X})=\{e\}$ and 
 $\hat{X}' \in \Vfam$  
 with $\delta^-_{F}(\hat{X}')=\{e'\}$.
 Note that $v \not\in X^+ \cup (X')^+$.
 We have $\hat{X}\cap \hat{X}', \hat{X}\cup \hat{X}'\in \Vfam$
 because $\Vfam$ is a ring-family.
 $u \in X\setminus X'$ and $u' \in X' \setminus X$ hold,
 and hence $e,e' \not\in \delta_{F}^-(\hat{X}\cap\hat{X}')$ holds.
 However, this means that $\delta_{F}^-(\hat{X}\cap\hat{X}')$ contains
 an edge distinct from $e$ and $e'$,
 and that this edge covers $\hat{X}$ or $\hat{X}'$.
 This contradicts the definition of $\hat{X}$ or $\hat{X}'$.

 We can also see that
 $F$ contains at most one edge entering $v$.
 To the contrary, suppose that
 there are two edges $f=uv$ and $f'=u'v$ in $F$.
 There exist $\hat{Y} \in \Vfam$  
 with $\delta^-_{F}(\hat{Y})=\{f\}$ and 
 $\hat{Y}' \in \Vfam$  
 with $\delta^-_{F}(\hat{Y}')=\{f'\}$ by the minimality of
 $F$.
 Note that $v \in Y \cap Y'$.
 We have $\hat{Y}\cap \hat{Y}', \hat{Y}\cup \hat{Y}'\in \Vfam$.
 If $f$ covers $\hat{Y}\cup \hat{Y}'$,
 then it covers $\hat{Y}'$ as well, which is a contradiction.
 Hence $f$ does not cover $\hat{Y}\cup \hat{Y}'$.
 Similarly, we can see that $f'$ does not cover $\hat{Y}\cup
 \hat{Y}'$, which means that $\delta^-_{F}(\hat{Y}\cup \hat{Y}')$ contains
 an edge that is distinct from $f$ and $f'$, and it
  covers $\hat{Y}$ or $\hat{Y}'$.
 However, this contradicts the definition of $\hat{Y}$ or $\hat{Y}'$.
\end{proof}

\section{LP relaxation for prize-collecting augmentation problem}
\label{sec.LP}

In this section, we present an LP relaxation for the prize-collecting
augmentation problem. 
Henceforth, we let $k$ denote the target connectivity from now on; 
The connectivity of each demand pair is $k-1$ in $(V,E_0)$,
and the problem
requires an increase in the connectivity of each demand pair by at least one.

For an edge $uv\in A$,
let $\Psi^{uv}$ denote the set of pairs $(j,j') \in W \times W$
such that 
$\psi^{uv}(j,j') = \true$.
A natural integer programming (IP) formulation for the prize-collecting biset covering
problem can be given by preparing 
variables $x(uv,j,j') \in \{0,1\}$ for each $uv \in A$ and $(j,j') \in \Psi^{uv}$, 
$x(v,j) \in \{0,1\}$ for each $v \in V$ and $j \in W$, 
and $y(i) \in \{0,1\}$ for each $i \in [d]$.
$x(uv,j,j')=1$ indicates that $uv$ is activated 
by weights $w$ with $w(u)=j$ and $w(v)=j'$.
$x(v,j)$ is equal to 1 if $v$ is assigned the weight $j$, and $0$ otherwise.
$y(i)$ indicates whether the connectivity requirement for $\{s_i,t_i\}$
is satisfied, and
$y(i)=0$ holds when
all bisets separating $\{s_i,t_i\}$ are covered.
The connectivity constraints require that, for each $i \in [d]$ and 
$\hat{X} \in \Vfam_i$, 
$y(i)=1$ holds or $\hat{X}$ is
covered by an activated edge, which is represented by 
$\sum_{uv \in \delta^-_A(\hat{X})}\sum_{(j,j') \in \Psi^{uv}} x(uv,j,j') +y(i)\geq 1$.
If $x(uv,j,j')=1$, then $u$ and $v$ must be assigned the weights $j$ and $j'$,
respectively.
This is represented by $x(u,j) \geq x(uv,j,j')$ and 
$x(v,j')\geq x(uv,j,j')$ for each $uv \in A$ and $(j,j') \in \Psi^{uv}$.
The objective is to minimize $\sum_{v \in V}\sum_{j \in W}j \cdot x(v,j) +
\sum_{i \in [d]}\pi_i \cdot y(i)$.
In conclusion, IP can be described as follows:
\begin{align}
&\mbox{minimize} && \sum_{v \in V}\sum_{j \in W} j \cdot x(v,j)
 +\sum_{i\in [d]}\pi_i \cdot y(i)\notag\\
&\mbox{subject to} &&
\sum_{uv \in\delta^-_A(\hat{X})} \sum_{(j,j') \in \Psi^{uv}} x(uv,j,j') 
+ y(i)\geq 1
&&\mbox{for $i\in [d]$, $\hat{X} \in \Vfam_i$,}\notag\\ 
&&&
x(u,j) \geq x(uv,j,j')  &&\mbox{for $uv \in A$, $(j,j') \in \Psi^{uv}$,}\label{ip.c2}\\
 &&& 
 x(v,j')\geq x(uv,j,j')  &&\mbox{for $uv \in A$, $(j,j') \in \Psi^{uv}$,}\label{ip.c3}\\
&&& x(v,j) \in \{0,1\} &&\mbox{for $v \in V$, $j\in W$,}\notag \\
&&& x(uv,j,j') \in \{0,1\} &&\mbox{for $uv \in A$, $(j,j') \in \Psi^{uv}$,}\notag\\
&&& y(i) \in \{0,1\} &&\mbox{for $i \in [d]$.}\notag
\end{align}
However, the LP relaxation obtained by dropping off the integrality constraints
from this IP has an unbounded integrality gap as follows. Consider the case where 
$d=1$, $\Vfam_1$
consists of only one biset $\hat{X}$, and 
$\delta_E(\hat{X})$ contains
$m$ edges incident to a node $u \in V\setminus X^+$.
Moreover, $W=\{0,1\}$ and each
edge $uv$ is activated by weights $w(u)=1$ and $w(v)=0$.
Suppose $\pi_1=+\infty$ so that $y(1)=0$ holds in any optimal
solutions for the IP and LP relaxation. 
For this instance, an integral solution activates one edge from
$\delta^-_A(\hat{X})$ by assigning weight 1 to $u$ and weight 0 to the
other end-node of the
chosen edge, which achieves the objective value 1.
On the other hand, define a fractional
solution $x$ so that
$x(u,1)=1/m$, $x(v,0)=1/m$,
and $x(uv,1,0)=1/m$
for all $uv \in \delta^-_A(\hat{X})$,
and the other variables are equal to 0.
This solution
is feasible for the LP relaxation,
and its objective value is $1/m$. This example implies that
the integrality gap of the LP relaxation is at least $m$.

For this reason, we need another LP relaxation.
Our idea is to strengthen \eqref{ip.c2} and \eqref{ip.c3}.
In the above IP, 
$x(u,j)$ is bounded by $x(uv,j,j')$ from below in \eqref{ip.c2}.
Instead, our new constraints bound $x(u,j)$
by $\sum_{v \in X: uv \in X}\sum_{j'\in W:(j,j')\in \Psi^{uv}}x(uv,j,j')$ for each $\hat{X} \in \Vfam$ with
$u \not\in X^+$.
However, these constraints are so strong that solutions feasible
to the prize-collecting biset covering problem do not satisfy it.
To remedy this drawback, we introduce new variables $x(uv,j,j',\hat{C})$
for each $\hat{C} \in \Mfam_{\Vfam}$ to replace $x(uv,j,j')$.
 $x(uv,j,j',\hat{C})$ is used for covering $\hat{X} \in \Vfam(\hat{C})$.
For each $\hat{C} \in \Mfam_{\Vfam}$, $\hat{X} \in \Vfam(\hat{C})$, $u \in
V\setminus X^+$, and $j \in W$,
$x(u,j)$ is bounded by $\sum_{v \in
X:uv\in A}\sum_{j'\in W:(j,j')\in \Psi^{uv}}x(uv,j,j',\hat{C})$.
\eqref{ip.c3} is similarly modified.
Summarizing, the following is the proposed LP relaxation.
\begin{align}
 &{\LP}(\Vfam)=\hspace*{-2em} \notag\\
 &\mbox{minimize} && \sum_{v \in V}\sum_{j \in W}j \cdot x(v,j)
 +\sum_{i\in [d]}\pi_i \cdot y(i)\notag\\
&\mbox{subject to} &&
\sum_{uv \in \delta^-_A(\hat{X})} \sum_{(j,j')\in \Psi^{uv}}x(uv,j,j',\hat{C}) + y(i)\geq 1
&&\mbox{for $\hat{C} \in \Mfam_{\Vfam}$, $i\in [d]$, $\hat{X} \in \Vfam_i(\hat{C})$,}\label{lp-c1}\\
&&&
x(u,j) \geq \sum_{\stackrel{v \in X:}{uv \in A}}\sum_{\stackrel{j' \in W:}{(j,j')\in \Psi^{uv}}}x(uv,j,j',\hat{C}) &&
\mbox{for $\hat{C} \in \Mfam_{\Vfam}$, $\hat{X}\in \Vfam(\hat{C})$, $u \in V\setminus X^+$, $j \in W$,}\label{lp-c2}\\
&&&
 x(v,j') \geq 
 \sum_{\stackrel{u \in V\setminus X^+:}{uv\in A}}\sum_{\stackrel{j\in W:}{(j,j')\in\Psi^{uv}}}x(uv,j,j',\hat{C}) &&
\mbox{for $\hat{C}\in \Mfam_{\Vfam}$, $\hat{X}\in \Vfam(\hat{C})$, $v \in X$, $j' \in W$, }\label{lp-c3}\\
&&& x(v,j) \geq 0 &&\mbox{for $v \in V$, $j \in W$,}\notag \\
&&& x(uv,j,j',\hat{C}) \geq 0 &&\mbox{for $uv \in A$, $(j,j') \in \Psi^{uv}$, $\hat{C} \in \Mfam_{\Vfam}$,}\notag\\
&&& y(i) \geq 0 &&\mbox{for $i \in [d]$.}\notag
\end{align}

{\bf Note:}
In \cite{Fukunaga14}, the author applied a similar idea of lifting LP relaxations for solving several covering problems
in edge- and node-weighted graphs. He defined a new LP relaxation by 
replacing edge variables by variables corresponding to pairs of edges
and constraints, and showed that the new LP relaxation has better integrality gap than the original one. 
This idea cannot be applied to the SNDP and the network activation problem straightforwardly because
they have an exponential number of constraints. Hence we instead define a new variable for each pair of edges and
min-cores, which makes the number of new variables being polynomial.

\begin{lemma}\label{lem.relaxation}
 $\LP(\Vfam)$ is at most the optimal value of the prize-collecting
 biset covering problem when $\Vfam$ is uncrossable.
 \end{lemma}
\begin{proof}
 Let $w\colon V \rightarrow W$ be a solution to the prize-collecting
 biset covering problem, and 
 let $A_w$ be the set of directed edges obtained by replacing each
 $\{u,v\}\in E_w$ with $uv$ and $vu$.
 For each $\hat{C} \in \Mfam_{\Vfam}$,
 let $A_{\hat{C}}$ be a minimal subset of $A_w$
 covering each $\hat{X} \in \Vfam(\hat{C})$ that is covered by $E_w$.
 We define an integer solution $(x,y)$ to $\LP(\Lfam)$ as follows:
 \begin{eqnarray*}
 y(i) &=& \begin{cases}
	1 & \mbox{if all bisets in $\Vfam_i$ are not covered by $E_w$,}\\
	0 & \mbox{otherwise,}
       \end{cases}\\
 x(uv,j,j',\hat{C})&=&
 \begin{cases}
	1 & \mbox{if $uv \in A_{\hat{C}}$ and $(j,j')=(w(u),w(v))$,}\\
	0 & \mbox{otherwise,}
       \end{cases}\\
 x(v,j) &=& 
 \begin{cases}
	1 & \mbox{if $j=w(v)$,}\\
	0 & \mbox{otherwise.}
       \end{cases}
 \end{eqnarray*}

We can see that the objective value of $(x,y)$ is at
most that of $w$. We here prove that $(x,y)$ is
 feasible for $\LP(\Vfam)$.
 Since $A_{\hat{C}}$ covers each $\hat{X} \in \Vfam_i(\hat{C})$ unless $y(i)=1$,
 we can see that \eqref{lp-c1} holds.
 By Lemma~\ref{lem.uncrossable}, 
 $\Vfam(\hat{C})$ is a ring-family.
 Hence, the right-hand side of \eqref{lp-c2} is at most one by
 Lemma~\ref{lem.ring-cover}.
 If it is one, then the left-hand side of \eqref{lp-c2} is also one 
 by the definition of $x$. Hence, $x$ satisfies \eqref{lp-c2}.
 It can be similarly observed from Lemma~\ref{lem.ring-cover}
 that $x$ satisfies \eqref{lp-c3}.
\end{proof}

In our algorithm, we first solve $\LP(\Vfam)$.
This is possible by the ellipsoid method 
under the assumption that 
a polynomial-time algorithm is available for computing
a minimal biset, including a specified node in its inner-part
over a ring-family.
This is
because the separation over the feasible region of $\LP(\Vfam)$ can be
done in polynomial time as follows.
The separation of \eqref{lp-c1} can be reduced to the submodular
function minimization problem for which polynomial-time algorithms are
known.
\eqref{lp-c2} has an exponential number of constraints for fixed $\hat{C}\in
\Mfam_{\Vfam}$, $u \in V$ and $j \in W$,
but a maximal biset in $\Vfam(\Cfam)$ such that $u \in V\setminus X^+$ is
unique and can be found in polynomial time by the above assumption and
from the fact that $\Vfam(\hat{C})$ is a ring-family.
Hence, it is sufficient to check a polynomial number of inequalities for the
separation of \eqref{lp-c2}, which can be done in
polynomial time.
The separation of \eqref{lp-c3} can be done similarly.
If $\Vfam$ is defined as $\bigcup_{i\in [d]}\Vfam_i^{\rm edge}$ or
$\bigcup_{i\in [d]}\Vfam_i^{\rm ele}$,
then the algorithm in the assumption is available, and the minimal biset can
be computed from maximum flows.
The separation of \eqref{lp-c1} can be done by the maximum flow
computation as well in such a case.
Moreover, $\LP(\Vfam)$ has a compact representation if
$\Vfam$ is $\bigcup_{i \in [d]}\Vfam_i^{\rm edge}$ or $\bigcup_{i \in
[d]}\Vfam_i^{\rm ele}$, 
and hence we can also use other LP solvers for solving
$\LP(\Vfam)$.

After solving $\LP(\Vfam)$, we round
each variable $y(i)$, $i \in [d]$ 
in the optimal solution
to either $0$ or $1$.
The demand pair $\{s_i,t_i\}$ is thrown away if $y(i)$ is rounded to
$1$.
We let $\NPCLP(\Vfam)$ denote the LP such that 
$y(i)$ is fixed to $0$ for all $i \in [d]$.
We then apply a primal-dual algorithm, given in the subsequent section, that
computes a spider for the remaining demand pairs.
The algorithm does not deal with $\NPCLP(\Vfam)$ directly but
runs on a simpler LP, which we call $\CoreLP(\Vfam)$.
The following is a description of $\CoreLP(\Vfam)$.
\begin{align}
&\CoreLP(\Vfam)= \hspace*{-4em} \notag\\
 &\mbox{minimize} && \sum_{v \in V}\sum_{j \in W}j \cdot (x_{\rm in}(v,j)+x_{\rm out}(v,j))\notag\\
&\mbox{subject to} &&
\sum_{uv \in \delta^-_{A}(\hat{X})} \sum_{(j,j')\in\Psi^{uv}}x(uv,j,j') \geq 1
&&\mbox{for $\hat{X} \in \Vfam$,}\label{label.primal-c1}\\
&&&
x_{\rm out}(u,j) \geq \sum_{\stackrel{v \in X:}{uv\in A}}\sum_{\stackrel{j' \in W:}{(j,j')\in\Psi^{uv}}}x(uv,j,j') &&
\mbox{for $\hat{X}\in \Vfam$, $u \in V\setminus X^+$, $j \in W$,
 }\label{primal-c2}\\
&&&
x_{\rm in}(v,j') \geq 
 \sum_{\stackrel{u \in V\setminus X^+:}{uv\in A}} \sum_{\stackrel{j \in W:}{(j,j')\in \Psi^{uv}}}x(uv,j,j') &&
\mbox{for $\hat{X}\in \Vfam$, $v \in X$, $j' \in W$, }\label{primal-c3}\\
&&& x_{\rm in}(v,j) \geq 0 &&\mbox{for $v \in V$, $j \in W$,}\notag \\
&&& x_{\rm out}(v,j) \geq 0 &&\mbox{for $v \in V$, $j \in W$,}\notag \\
&&& x(uv,j,j') \geq 0 &&\mbox{for $uv \in A$, $(j,j') \in \Psi^{uv}$.}\notag
\end{align}

Instead of $x(v,j)$ in $\LP(\Vfam)$, $\CoreLP(\Vfam)$ has two variables $x_{\rm in}(v,j)$ and
$x_{\rm out}(v,j)$ for each pair of $v \in V$ and $j \in W$, where
$x_{\rm in}(v,j)$ indicates if $v$ is assigned the weight $j$ for
activating edges entering $v$, and 
$x_{\rm out}(v,j)$ indicates if $v$ is assigned the weight $j$ for
activating edges leaving $v$.
We require this modification in order to obtain a primal-dual algorithm.

$\CoreLP(\Vfam)$ does not relax $\NPCLP(\Vfam)$
or the biset covering problem.
In fact, the analysis of our primal-dual algorithm does not use $\CoreLP(\Vfam)$.
The LP relaxation we use is $\CoreLP(\Lfam)$ defined from some
subfamily $\Lfam$ of $\Vfam$.
We do not know $\Lfam$ beforehand, but we can show that
$\Lfam$ is a strongly laminar family of cores of $\Vfam$.
The following lemma indicates that in this case
$\CoreLP(\Lfam)$ is within a constant factor of $\NPCLP(\Vfam)$.

\begin{lemma}\label{lem.corelpvslp}
 $\CoreLP(\Lfam) \leq 2\NPCLP(\Vfam)$ 
 if $\Vfam$ is uncrossable and $\Lfam$ is a strongly laminar family of
 cores of $\Vfam$.
 \end{lemma}
 \begin{proof}
  Let $x$ be an optimal solution for $\NPCLP(\Vfam)$.
  Decreasing $x$ greedily,
  we transform $x$ into a minimal feasible solution to $\NPCLP(\Lfam)$.
  Then, we define a solution $x'$ to $\CoreLP(\Lfam)$
  so that $x'(uv,j,j')=\max_{\hat{C} \in \Mfam_{\Vfam}}x(uv,j,j',\hat{C})$ for each $uv \in A$ and
  $(j,j')\in \Psi^{uv}$, and
  $x'_{\rm out}(v,j)=x'_{\rm in}(v,j)=x(v,j)$ for each $v \in V$ and
  $j\in W$. The objective value of $x'$ in $\CoreLP(\Lfam)$
  is at most $2\NPCLP(\Vfam)$.
  Hence, it suffices to prove that $x'$ is feasible to $\CoreLP(\Lfam)$.

  \eqref{label.primal-c1} follows from \eqref{lp-c1}.
  Let $\hat{C} \in \Mfam_{\Vfam}$.
  If \eqref{primal-c2} is violated for $\hat{X} \in \Lfam(\hat{C})$, $u \in
  V\setminus X^+$ and $j\in W$,
  then there exists a pair of $uv \in \delta^-_A(\hat{X})$
  and $\hat{C}' \in \Mfam_{\Vfam}$ such that 
  $x(uv,j,j',\hat{C}') > x(uv,j,j',\hat{C})$.
  The minimality of $x$ implies that 
  there exists $\hat{Y} \in \Lfam(\hat{C}')$
  with $uv \in \delta^-_A(\hat{Y})$.
  The strong laminarity of $\Lfam$ indicates that 
  $\hat{Y}$ is comparable with $\hat{X}$, but this 
  means that $\hat{Y} \in \Lfam(\hat{C})$, which is a contradiction
  because 
  a core does not include two min-cores.
  Therefore, $x'$ satisfies \eqref{primal-c2}.
  We can similarly prove that $x'$ satisfies \eqref{primal-c3} as well.
 \end{proof}

The dual of $\CoreLP(\Vfam)$ is 
\begin{align}
&\CoreDual(\Vfam)= \hspace*{-5em} \notag \\
 &\mbox{maximize}&&\sum_{\hat{X} \in \Vfam}z(\hat{X})\notag\\
 &\mbox{subject to} &&
 \sum_{\hat{X} \in \Vfam: uv \in \delta^-_{A}(\hat{X})} 
 z(\hat{X}) 
 \leq 
\sum_{\hat{X}\in \Vfam: uv \in \delta^-_{A}(\hat{X})}
&&\hspace*{-3em}\left(z(\hat{X},u,j) + z(\hat{X},v,j')\right) \notag\\
&&&&&\mbox{for $uv\in A$, $(j,j') \in \Psi^{uv}$,}\label{dual-c1}\\
&&& \sum_{\hat{X}\in \Vfam:v \in X}z(\hat{X},v,j') \leq j'
&&\mbox{for $v\in V$, $j' \in W$,}\label{dual-c3}\\
&&& \sum_{\hat{X}\in \Vfam:u \in V \setminus X^+}z(\hat{X},u,j) 
\leq j
&&\mbox{for $u\in V$, $j \in W$,}\label{dual-c3'}\\
 &&& z(\hat{X})\geq 0 
&&\mbox{for $\hat{X}\in \Vfam$,}\notag\\
 &&& z(\hat{X},v,j)\geq 0 
&&\mbox{for $\hat{X}\in \Vfam$, $v \not\in \Gamma(\hat{X})$, $j \in W$.}\notag
 \end{align}

In the subsequent section, we present an algorithm for computing
node weights activating a spider 
and a solution $z$ feasible to $\CoreDual(\Lfam)$
for some strongly laminar family $\Lfam$ of cores.
The sum of weights  is bounded in terms
of $\sum_{\hat{X}\in \Lfam}z(\hat{X})$.

\section{Primal-dual algorithm for computing spiders}
\label{sec.primal-dual}

A spider for a biset family $\Vfam$ is an edge set $S\subseteq
E$ 
such that there exist $h \in V$ and $\hat{X}_1,\ldots,\hat{X}_f \in
\Mfam$, and $S$ can be decomposed into subsets $S_1,\ldots,S_f$
that satisfy the following conditions:
\begin{itemize}
 \item $V(S_i) \cap V(S_j)\subseteq \{h\}$ for each $i,j \in [f]$ with
       $i\neq j$;
 \item $S_i$ covers $\Cfam(\hat{X}_i,h)$ for each $i \in [f]$;
 \item if $f=1$, then $\Cfam(\hat{X}_1,h)=\Cfam(\hat{X}_1)$;
 \item $h \not\in X^+_i$ for each $i \in [f]$.
\end{itemize}
$h$ is called the head, and $\hat{X}_1,\ldots,\hat{X}_f$ are called the feet of the spider.
For a spider $S$, we let $f(S)$ denote the number of its feet.
Note that this definition of spiders for biset families
is slightly different from the original one
in~\cite{Nutov12uncrossable}, where
an edge set is a spider in~\cite{Nutov12uncrossable} even if it does not
satisfy the last condition given above.

In this section, we present an algorithm for computing spiders.
More precisely, we prove the following theorem.

\begin{theorem}\label{thm.spideralgorithm}
 Let $\Vfam$ be an uncrossable family of bisets.
 There exists a polynomial-time algorithm for computing $w\colon
 V\rightarrow W$
 and a strongly laminar family $\Lfam$ of cores
 such that
 $E_w$ contains a spider $S$ and 
 $w(V)/f(S) \leq \CoreLP(\Lfam)/|\Mfam_{\Vfam}|$ holds.
\end{theorem}

Our algorithm keeps 
an edge set $F\subseteq E$, 
core families $\Lfam, \Afam\subseteq \Cfam$,
and a feasible solution $z$ to $\CoreDual(\Lfam)$.
We initialize the dual variables $z$ to $0$ and $F$ to the empty set.
$\Lfam$ and $\Afam$ are initialized to the family $\Mfam$ of min-cores of $\Vfam$.
By Lemma~\ref{lem.uncrossable}, 
$\Lfam$ and $\Afam$ are pairwise strongly disjoint.
The algorithm always
maintains $\Lfam$ being strongly laminar and $\Afam$ being pairwise strongly disjoint.

{\bf Increase phase:}
After initialization, we increase dual variables
$z(\hat{X})$, $\hat{X}\in \Afam$ uniformly.
We introduce the concept of time. Each of the variables is
increased by one in a unit of time. 

For satisfying the constraints of $\CoreDual(\Lfam)$,
we have to increase other variables as well.
Let $uv \in \delta^-_A(\hat{X})$ and $(j,j') \in \Psi^{uv}$.
To satisfy \eqref{dual-c1},
for each such pair of $uv$ and $(j,j')$, 
we have to increase 
$z(\hat{X},u,j)$,
or $z(\hat{X},v,j')$.
Note that $z(\hat{X},u,j)$ 
is bounded from above by \eqref{dual-c3'} for $(u,j)$,
and $z(\hat{X},v,j')$ is bounded from above 
by \eqref{dual-c3} for $(v, j')$.
Our algorithm first increases $z(\hat{X},v,j')$ at the same speed
as $z(\hat{X})$ until 
\eqref{dual-c3} becomes tight for $(v,j')$.
Let $\tau(v,j')$ denote the time when 
\eqref{dual-c3} becomes tight for $(v,j')$.
After time $\tau(v,j')$,
the algorithm increases $z(\hat{X},u,j)$.
There may exist another pair of $uv' \in \delta^-_A(\hat{X})$ (possibly
$v'=v$) and $(j,j'') \in \Psi^{uv'}$.
In this case,
we stop increasing $z(\hat{X},v',j'')$ 
at time $\tau(v,j')$
even if \eqref{dual-c3} is not yet
tight for $(v',j'')$ at time $\tau(v,j')$,
We say that $(uv,j,j')$ gets {\em tight}
when we cannot increase 
$z(\hat{X},u,j)$ 
or $z(\hat{X},v,j')$.

{\bf Events:}
After increasing the dual variables for some time, 
we encounter an event that the variable $z(\hat{X})$ for some $\hat{X} \in
\Afam$ can no longer be increased
because of a tight tuple $(uv,j,j')$ with $uv \in \delta_A^-(\hat{X})$ and $(j,j')
\in \Psi^{uv}$.
Let $\tilde{\tau}$ be the time when this event occurs.

It is possible that 
more than one such tuple may get simultaneously tight.
We choose an arbitrary pair of $u \in V \setminus X^+$ and $j \in W$
such that there exists a tight tuple $(uv,j,j')$ with $uv \in \delta^-_A(\hat{X})$
and $(j,j') \in \Psi^{uv}$.
Let $(uv_{1},j_1),\ldots,(uv_{p},j_p)$ be the pairs of edges leaving $u$ in
$\delta^-_A(\hat{X})$ and weights such that 
$(uv_{p'},j,j_{p'})$ is a tight tuple for each $p' \in [p]$.
For each $p' \in [p]$, define
$\hat{Y}_p'$ as the 
minimal core in $\Lfam$ such that $uv_{p'} \in \delta^-_A(\hat{Y}_{p'})$.
Without loss of generality, 
suppose $\hat{Y}_1 \subseteq \cdots \subseteq \hat{Y}_p \subseteq \hat{X}$.
We add the
undirected edge $\{u,v_{1}\}$ to $F$, and assign the weight $j$ to $u$
and weight $j_1$ to $v_1$.
We say that $\hat{X}$ is the {\em witness} of the edge $\{u,v_{1}\}$.
If $z(\hat{X}',u,j)>0$ for some biset $\hat{X}' \in \Lfam$ comparable with
$\hat{X}$, then 
$\hat{Y}_1 \subseteq \hat{X}' \subseteq \hat{X}$ holds
because
the algorithm does not increase
$z(\hat{X}',u,j)$
unless there exists a pair of $uv\in \delta^-_A(\hat{X}')$
and $(j,j')\in \Psi^{uv}$
such that 
\eqref{dual-c3} is tight for $(v,j')$, and $(uv,j,j')$ is tight when
\eqref{dual-c3'} tightens for $(u,j)$.

Let $B$ be the set of directed edges leaving $u$ whose corresponding
undirected edges are added to $F$ at time $\tilde{\tau}$ or earlier, where
$B$ does not contain $uv$ if $\{u,v\}$ is added to $F$ because of $vu$.
We define two cases here.
In Case~(a),  $|B|=1$ holds 
and 
there exists a core $\hat{Z} \in \Cfam$ such that 
$\hat{X} \subset\hat{Z}$ and $\hat{Z}$ is not covered by $F$.
In Case~(b), $|B|\geq 2$ holds or 
all cores $\hat{Z} \in \Cfam$ with $\hat{X} \subset
\hat{Z}$ are covered by $F$.

{\em Case}~(a): $|B|=1$
and 
there exists a core $\hat{Z} \in \Cfam$ such that 
$\hat{X} \subset\hat{Z}$ and $\hat{Z}$ is not covered by $F$.
Let $\hat{Z}$ be a minimal core among such cores. $\hat{Z}$ is
unique because $\Cfam_F(\hat{X})$ is a ring-family by
Lemmas~\ref{lem.uncrossable} and \ref{lem.residual-uncrossable}.
We add $\hat{Z}$ to both $\Lfam$ and $\Afam$, and remove
$\hat{X}$ from $\Afam$.

\begin{lemma}\label{lem.afam}
  $\Afam$ is the family of min-cores of
 $\Vfam_F$
 after the update of Case~{\rm (}a\/{\rm )}.
\end{lemma}
\begin{proof}
 Let $uv \in B$.
 Recall that $uv$ covers $\hat{X}$, and hence $v \in X \subseteq Z$.
 It suffices to show that $\{u,v\}$ covers no core in $\Afam$.
 Let $\hat{Z}' \in \Afam$. If $\hat{Z}'=\hat{Z}$, then its definition
 implies that $\{u,v\}$ does not cover it.
 Hence, suppose that $\hat{Z}'\neq \hat{Z}$.
 Let $F'$ represent $F$ before $\{u,v\}$ is added.
 Since $\hat{Z}'$ was in $\Afam$ before the update,
 $\hat{Z}'$ is a min-core of  $\Vfam_{F'}$, which 
 implies that $\hat{Z}$ and $\hat{Z}'$ are strongly disjoint by
 Lemma~\ref{lem.uncrossable} (iii).
 $v \not\in Z'$ follows from $v \in Z$.
 Since 
 $\{u,v\}$ does not cover $\hat{Z}$, we have $u \in Z^+$, and 
 hence $u \not\in Z'$.
 These indicate that $\{u,v\}$ does not cover $\hat{Z}'$.
\end{proof}

Lemma~\ref{lem.afam} indicates that
$\Afam$ is pairwise strongly disjoint and $\Lfam$ is strongly laminar
even after the update.

{\em Case} (b): $|B|\geq 2$ or 
all cores $\hat{Z}$ with $\hat{X} \subset
\hat{Z}$ are covered by $F$.
In this case, we go to the deletion phase, which
removes several edges from $F$.
We then output the obtained edge set, 
node weights for activating the edge set,
and $\Lfam$.
We will show that the edge set is a spider with $|B|$ feet.

{\bf Deletion Phase}\/:
Let $\hat{Y} \in \Afam$, and let
$\hat{Y}_1,\ldots,\hat{Y}_{l-1}$ be the cores included by $\hat{Y}$
in $\Lfam$.
We also let $\hat{Y}_l=\hat{Y}$.
We assume without loss of generality that $\hat{Y}_1 \subset \cdots \subset
\hat{Y}_{l}$ holds.
$\hat{Y}_1$ is a min-core of $\Vfam$.
Let $F_{\hat{Y}}$ be the edges in $F$ whose witnesses are in 
$\{\hat{Y}_1,\ldots,\hat{Y}_{l}\}$.
Note that $F$ can be partitioned into
$F_{\hat{Y}}$, $\hat{Y}\in \Afam$.

For each $l' \in [l]$,
$F$ contains an edge
$\{u_{l'},v_{l'}\}$ whose witness is $\hat{Y}_{l'}$.
Without loss of generality, we have $v_{l'} \in Y_{l'}$ and $u_{l'} \in
Y^+_{l'+1}\setminus Y^+_{l'}$ for $l' \in [l]$,
where we let $Y^+_{l+1}=V$ for convenience.
We apply the following algorithm to delete several edges from $F_{\hat{Y}}$.

\begin{description}\setlength{\itemsep}{0pt}
 \item[Deletion algorithm]
 \item[Step 1:]
	    Define $p$ as $l$ and $S_{\hat{Y}}$ as $F_{\hat{Y}}$.
 \item[Step 2:]
	    Let $q$ be the smallest integer in $[p]$ such that 
	    $v_{p} \in Y_{q}$.
	    Remove $\{u_{p-1},v_{p-1}\}, \ldots,\{u_{q},v_{q}\}$
	    from $S_{\hat{Y}}$.
 \item[Step 3:] If $q>1$, then set $p$ to $q-1$ and go back to Step~2.
	    Otherwise, output $S_{\hat{Y}}$ and terminate.
\end{description}

\begin{figure}[th]
 \centering
 \begin{tikzpicture}[line width=1.5pt]
  \tikzstyle{every label}=[label distance=-3pt]
  \draw[biset] (-2.4,-2.4) rectangle (5.4,3.4);
  \draw[biset outer] (-2.5,-2.5) rectangle (5.5,3.5);
  \draw[biset] (-1.8,-1.8) rectangle (4.3,2.8);
  \draw[biset outer] (-1.9,-1.9) rectangle (4.4,2.9);
  \draw[biset] (-1.2,-1.2) rectangle (3.2,2.2);
  \draw[biset outer] (-1.3,-1.3) rectangle (3.3,2.3);
  \draw[biset] (-.6,-.6) rectangle (2.1,1.6);
  \draw[biset outer] (-.7,-.7) rectangle (2.2,1.7);
  \draw[biset] (0,0) rectangle (1,1);
  \draw[biset outer] (-.1,-.1) rectangle (1.1,1.1);

  \node[vertex, label=north:$u_5$] (v1) at (6.5,.5){};
  \node[vertex, label={[label distance=-4pt]north west:$u_4$}] (v2) at (5.3,1.2){};
  \node[vertex, label=north:$v_4$] (v2n) at (3.8,1.2){};
  \node[vertex, label={[label distance=-7pt]south west:$v_5$}] (v3) at (3.2,.5){};
  \node[vertex, label=south:$u_3$] (v3n) at (3.7,0){};
  \node[vertex, label=north:$u_2$] (v4) at (2.6,.7) {};
  \node[vertex, label=south:$v_3$] (v5) at (1.45,0) {};
  \node[vertex, label={[label distance=-3pt]east:$u_1$}] (v5n) at (2.05,.35) {};
  \node[vertex, label={[label distance=-7pt]north east:$v_2$}] (v6) at (.5,.7) {};
  \node[vertex, label={[label distance=-5pt]west:$v_1$}] (v6n) at (.6,.35) {};

  \node at (5,-2) {$\hat{Y}_5$};
  \node at (3.9,-1.4) {$\hat{Y}_4$};
  \node at (2.8,-.8) {$\hat{Y}_3$};
  \node at (-.25,-.2) {$\hat{Y}_2$};
  \node at (.13,.83) {$\hat{Y}_1$};
  \draw[red] (v1) -- (v3);
  \draw (v2) -- (v2n);
  \draw[red] (v3n) -- (v5);
  \draw (v4) -- (v6);
  \draw[red] (v5n) -- (v6n);
 \end{tikzpicture}
 \caption{An example of $\hat{Y}_1,\ldots,\hat{Y}_l$ and
 $\{u_1,v_1\},\ldots,\{u_l,v_l\}$ with $l=5$. Red edges are those chosen
 in $S_{\hat{Y}}$. Areas surrounded by the dotted lines represent bisets,
 and dark gray areas represent boundaries of bisets.}
 \label{fig.deletion}
 \end{figure}
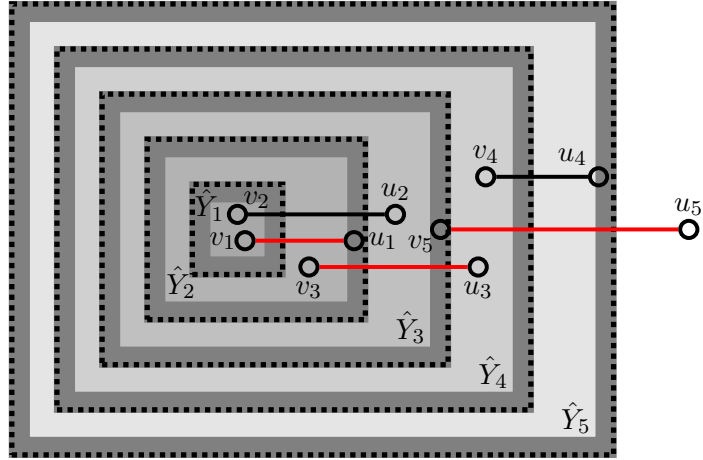

Figure~\ref{fig.deletion} illustrates an example to which the deletion
algorithm is applied.
Below, we let $S_{\hat{Y}}$ denote the edge set obtained by applying
the deletion algorithm to $F_{\hat{Y}}$.

\begin{lemma}\label{lem.fy}
Any core in $\Cfam(\hat{Y}_1,u_{l})$
is covered by at least one edge in $S_{\hat{Y}}$.
The core $\hat{Y}_{l'}$ is covered by exactly one edge in $S_{\hat{Y}}$ for
 each $l' \in [l]$.
\end{lemma}
\begin{proof}
 Let $l' \in [l]$.
 First, we show that $\hat{Y}_{l'}$ is covered by exactly one edge in 
 $S_{\hat{Y}}$.
 When
 the event occurs to $\hat{Y}_{l'}$, the algorithm adds the edge
 $\{u_{l'},v_{l'}\}$
 covering $\hat{Y}_{l'}$ to $F$, 
 and defines $\hat{Y}_{l'}$ as the witness of the edge.
 $\{u_{l'},v_{l'}\}$ is not removed by the deletion algorithm
 unless another edge covering $\hat{Y}_{l'}$ remains in $S_{\hat{Y}}$.
 Hence $\hat{Y}_{l'}$ is covered by at least one edge after applying
 the deletion algorithm.
 Let $p$ be the minimum integer in $[l]$ such that
$\{u_p,v_p\} \in S_{\hat{Y}}$ covers $\hat{Y}_{l'}$.
 By way of constructing $\Lfam$, we have $p \geq l'$.
 Suppose that another edge $\{u_{p'},v_{p'}\} \in S_{\hat{Y}}$
 covers $\hat{Y}_{l'}$ as well.
 Then, $v_{p'} \in Y_{l'}$ holds.
 The definition of $p$ indicates that $p' > p$.
 However, in this case, the deletion algorithm removes
 $\{u_p,v_p\}$ from $S_{\hat{Y}}$.
 Hence, $\hat{Y}_{l'}$ is covered by exactly one edge in $S_{\hat{Y}}$.

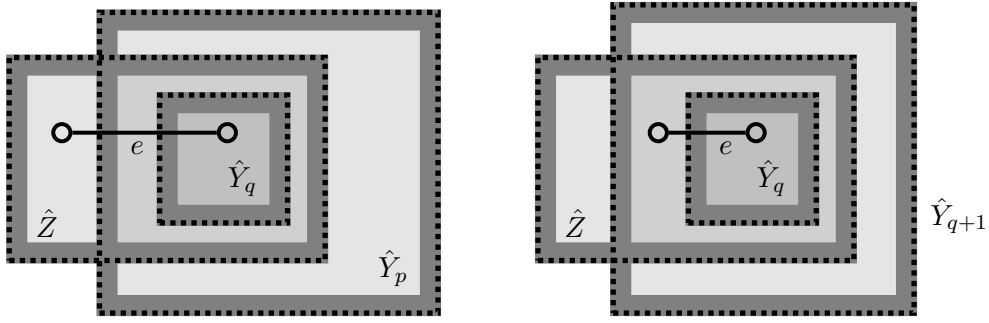
\begin{figure}[t]
 \centering

\begin{tikzpicture}[line width=1.5pt]

 \draw[biset] (1.7,-.7) rectangle +(4.3,3.8);
 \draw[biset] (.5,0) rectangle +(4,2.5);
 \draw[biset] (2.5,.5) rectangle +(1.5,1.5);
 \draw[biset outer] (1.6,-.8) rectangle +(4.5,4);
 \draw[biset outer] (.4,-.1) rectangle +(4.2,2.7);
 \draw[biset outer] (2.4,.4) rectangle +(1.7,1.7);

 \node at (.9,.4) {$\hat{Z}$};
 \node at (3.5,1) {$\hat{Y}_q$};
 \node at (5.5,-.2) {$\hat{Y}_{p}$};

 \node[vertex](a) at (1.1,1.6) {};
 \node[vertex](b) at (3.3,1.6) {};
 \draw (a) -- (b);
 \node at (2.1,1.4) {$e$};

 \begin{scope}[xshift=200pt]
 \draw[biset] (.5,0) rectangle +(4,2.5);
 \draw[biset] (1.5,-.7) rectangle +(3.8,3.9);
 \draw[biset] (2.5,.5) rectangle +(1.5,1.5);
 \draw[biset outer] (.4,-.1) rectangle +(4.2,2.7);
 \draw[biset outer] (1.4,-.8) rectangle +(4,4.1);
 \draw[biset outer] (2.4,.4) rectangle +(1.7,1.7);

 \node at (.9,.4) {$\hat{Z}$};
 \node at (3.5,1) {$\hat{Y}_q$};
 \node at (6,.5) {$\hat{Y}_{q+1}$};

 \node[vertex](c) at (2,1.6) {};
 \node[vertex](d) at (3.3,1.6) {};
 \draw (c) -- (d);
 \node at (2.9,1.4) {$e$};

 \end{scope}
\end{tikzpicture}

 \caption{Bisets in the proof of Lemma~\ref{lem.fy}. The left figure
illustrates the case where $p > q$, and the right figure illustrates the
 case where $p=q$.}
 \label{fig.lemmafy}
\end{figure}

 Let $\hat{Z} \in \Cfam(\hat{Y}_1,u_{l})$.
 We show that $\hat{Z}$ is covered by at least one edge in $S_{\hat{Y}}$.
  To the contrary,
 suppose that $\hat{Z}$  is 
 covered by no edge in $S_{\hat{Y}}$.
 Let $\hat{Z}$ be a maximal core among such cores,
 and let $q$ be the maximum integer in $[l]$ such that 
 $\hat{Y}_q \subseteq \hat{Z}$.
 By the above claim, $S_{\hat{Y}}$ contains the edge
 $e=\{u_{p},v_{p}\}$ covering $\hat{Y}_q$.
 Since $e$ does not cover $\hat{Z}$, we have $e \subseteq Z^+$, and
 $p < l$ holds because $u_{l}\not\in Z^+$.

 Suppose that $p>q$.
 The left example in Figure~\ref{fig.lemmafy}
 illustrates this case.
 By the maximality of $q$, 
 $\hat{Y}_{p}$ is not included by $\hat{Z}$, and hence
 $\hat{Z} \subset \hat{Z}\cup \hat{Y}_{p}$ holds.
 Since $\hat{Z} \cup \hat{Y}_p \in \Cfam(\hat{Y}_1,u_l)$,
 the maximality of $\hat{Z}$ indicates that
 $\hat{Z}\cup \hat{Y}_{p}$ is covered by an edge in $S_{\hat{Y}}$.
 Let
 $f$ be an edge in $S_{\hat{Y}}$ covering $\hat{Z}\cup \hat{Y}_{p}$.
 Since $e \subseteq Z^+$, $e$ does not cover
 $\hat{Z}\cup \hat{Y}_{p}$, implying $e \neq f$.
 $f$ covers $\hat{Z}$ or $\hat{Y}_{p}$.
 If $f$ covers $\hat{Y}_{p}$, then
 $\hat{Y}_p$ is covered by two edges in $S_{\hat{Y}}$, which is a contradiction.
 Hence, $f$ covers $\hat{Z}$, which is a contradiction again.

 Next, consider the case where $p=q$.
 The example on the right side of Figure~\ref{fig.lemmafy} illustrates this case.
 $e \subseteq Y^+_{q+1}$ follows from $p=q$.
 Hence, $e \subseteq Y^+_{q+1} \cap Z^+$, and $e$ does not cover
 $\hat{Y}_{q+1}\cap \hat{Z}$.
 By the maximality of $q$, $\hat{Y}_{q+1}$ is not included by $\hat{Z}$,
 and hence $\hat{Y}_{q+1}\cap \hat{Z}\subset \hat{Y}_{q+1}$.
 By Lemma~\ref{lem.afam}, $\hat{Y}_{q+1}$ was a minimal core in
 $\Cfam(\hat{Y}_1,u_{l})$
 that was not covered by $F$ when $e$ was added to $F$.
 Note that $\hat{Y}_{q+1}\cap \hat{Z} \in \Cfam(\hat{Y}_1,u_l)$.
 Hence, an edge in $F$ covered $\hat{Y}_{q+1}\cap \hat{Z}$
 when $e$ was added to $F$.
 Let $g$ denote such an edge.
 Since $g$ does not cover $\hat{Y}_{q+1}$,
 we have $g \subseteq Y^+_{q+1}$, implying that the witness of $g$ is
 included by $\hat{Y}_{q+1}$.
 $\hat{Y}_{q}$ is not the witness of $g$ because $e \neq g$.
 Hence, the witness of $g$ is also included by $\hat{Y}_{q}$.
 From this, it follows that $g \subseteq Y^+_{q} \subseteq Z^+$.
 However, it indicates that $g$ does not cover $\hat{Y}_{q+1} \cap \hat{Z}$,
 which is a contradiction.
 \end{proof}

Let $h$ be the node that each edge in $B$ leaves.
When $|B|\geq 2$, 
let $\mathcal{X}$ be the family of witnesses of edges in $B$.
We apply the deletion algorithm to each
$\hat{Y} \in \mathcal{X}$ to obtain $S_{\hat{Y}}$, and define $S=\bigcup_{\hat{Y}\in \mathcal{X}}S_{\hat{Y}}$.
When $|B|=1$,
let $\hat{X}$ be the witness of the edge in $B$,
and let $\mathcal{X}$ be the family of $\hat{X}$ and maximal cores in
$\Lfam \setminus \Afam$ that is not comparable with $\hat{X}$.
We apply the deletion algorithm to each core $\hat{Y}'\in \mathcal{X}$ to
obtain $S_{\hat{Y}'}$ and define $S=\bigcup_{\hat{Y}' \in \mathcal{X}}S_{\hat{Y}'}$
when $|B|=1$.
In the following lemmas, it will be shown that $S$ is a spider with 
$|B|$ feet and $h$ is the head of $S$.

\begin{lemma}\label{lem.spider1}
 When $|B|=1$,
 the edge set $S$
  is a spider with only one foot, and its head is $h$.
\end{lemma}
\begin{proof}
 Let $\hat{X}$ be the witness of the edge in $B$
 and $\hat{M}$ be the min-core included by $\hat{X}$.
 We prove that $S$ is a spider and its foot is $\hat{M}$.
 Lemma~\ref{lem.fy} indicates that all cores in $\Cfam(\hat{M},h)$ are
 covered by $S_{\hat{X}}$.
 Hence, it suffices to show that
 each core $\hat{Z} \in \Cfam(\hat{M})$ with $h \in Z^+$ is covered by
 $S$.
 Suppose that $\hat{Z}$ is covered by no edge in $S$.
 Let $\hat{Z}$ be the minimal core among such cores.
 There exists an edge $e=\{a,b\} \in F$ that covers $\hat{Z}$.
 Let $\hat{K}_1$ be the witness of $e$, and let $a \in K_1$ and
 $b\not\in K^+_1$, without loss of generality. If $\hat{K}_1\in \mathcal{X}$,
 then $e$ remains in $S$.
 Hence $\hat{K}_1\not\in \mathcal{X}$.
 $\hat{K}_1$ is either incomparable with
 $\hat{X}$ or is included by $\hat{X}$.
 If more than one edge in $F$ cover $\hat{Z}$ and one of them
 gives $\hat{K}_1$ incomparable with $\hat{X}$,
 then we choose such an edge as $e$.

 \begin{figure}[th]
  \centering
 \begin{tikzpicture}[line width=1.5pt]
  \draw[biset] (-1,-.5) rectangle (.5,.5);
  \draw[biset] (-.5,-1) rectangle +(3,2.4);
  \node at (-1.2,-.9) {$\hat{X}$};
  \node at (1.1,-.6) {$\hat{Z}$};

  \draw[biset] (1.8,-1.7) rectangle +(3,2.4);
  \draw[biset] (3.3,-1) rectangle +(1,1);
  \node at (3.9,-.6) {$\hat{K}_1$};
  \node at (5.3,-1.4) {$\hat{K}_2$};

  \draw[biset outer] (-1.1,-.6) rectangle (.6,.6);
  \draw[biset outer] (-.6,-1.1) rectangle +(3.2,2.6);
  \draw[biset outer] (1.7,-1.8) rectangle +(3.2,2.6);
  \draw[biset outer] (3.2,-1.1) rectangle +(1.2,1.2);

  \node[vertex] (a) at (3.6,-.3){};
  \node[vertex] (b) at (1.9,-.3){};
  \draw (a) -- (b);
  \node at (2.9,-.1) {$e$};

  \node[vertex] (c) at (3.8,.35){};
  \node[vertex] (d) at (3.8,1.4){};
  \draw (c) -- (d);
  \node at (4,1.1) {$f$};

  \node[vertex] (e) at (1.2,.35){};
  \node[vertex] (f) at (2.45,.35){};
  \draw (e) -- (f);
  \node at (2.2,.1) {$g$};

  \begin{scope}[xshift = 280pt]
   \draw[biset] (-2.5,-2) rectangle (1.5,2.5);
   \draw[biset] (-1.5,-1.3) rectangle (1,1.5);
   \draw[biset] (-.5,-.8) rectangle +(3,1.6);

   \draw[biset outer] (-2.6,-2.1) rectangle (1.6,2.6);
   \draw[biset outer] (-1.6,-1.4) rectangle (1.1,1.6);
   \draw[biset outer] (-.6,-.9) rectangle +(3.2,1.8);

   \node at (-2.1,2.1) {$\hat{X}$};
   \node at (-1,1.1) {$\hat{K}_1$};
   \node at (2.1,-.3) {$\hat{Z}$};
   
  \node[vertex] (a) at (-2,0){};
  \node[vertex] (b) at (.2,0){};
  \draw (a) -- (b);
  \node at (-1,-.2) {$e$};
  \end{scope}
 \end{tikzpicture}
  \caption{Bisets in the proof of Lemma~\ref{lem.spider1}. The left
  figure illustrates the case where $\hat{K}_1$ is incomparable with
  $\hat{X}$, and
  the right illustrates the case where $\hat{K}_1$ is included by $\hat{X}$.}
  \label{fig.spider1}
 \end{figure}
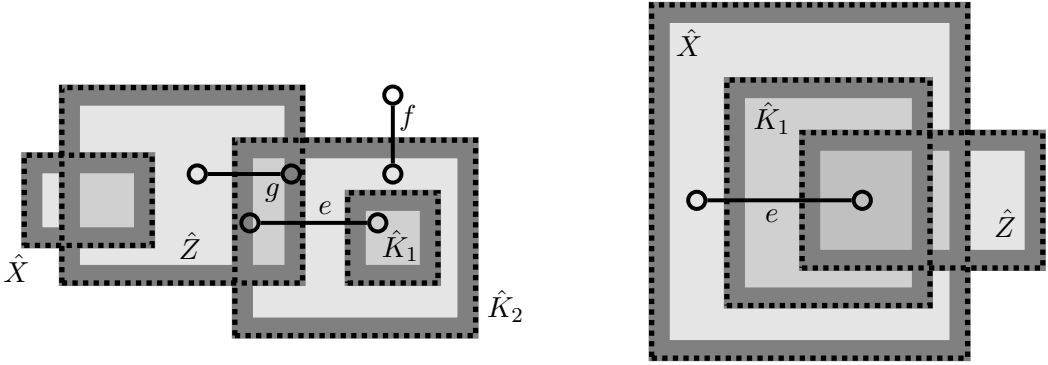

 Suppose that $\hat{K}_1$ is incomparable with $\hat{X}$.
 The left example in Figure~\ref{fig.spider1} illustrates this case.
 Let $\hat{K}_0$ be the min-core included by $\hat{K}_1$,
 and let $\hat{K}_2$ be the minimal core in $\Lfam$ with
 $\hat{K}_1\subset \hat{K}_2$.
 Note that $e \subseteq K^+_2$, and
 $\hat{Z}$ and $\hat{K}_2$ are incomparable.
 Then, $\hat{K}_2\setminus \hat{Z} \in \Cfam(\hat{K}_0)$ holds, and 
 it is covered by some edge $f\in S$ by Lemma~\ref{lem.fy}.
 Since $f$ does not cover $\hat{Z}$, it has one end-node in $K_2 \setminus
 Z^+$ and the other in $V \setminus (K^+_2 \cup Z)$.
 On the other hand, $\hat{Z}\setminus \hat{K}_2 \in \Cfam(\hat{M})$.
 The minimality of $\hat{Z}$ indicates that $\hat{Z}\setminus \hat{K}_2$
 is covered by some edge $g \in S$.
 Since $g$ does not cover $\hat{Z}$,
 it has one end-node in $Z\setminus K^+_2$ and the other in $K_2\cap Z^+$.
 These imply that $f\neq g$, and both $f$ and $g$ cover $\hat{K}_2$.
 If $\hat{K}_2 \in \Lfam\setminus \Afam$, then this is a contradiction because
 any core in $\Lfam\setminus \Afam$
 is covered by exactly one edge in $S$ by Lemma~\ref{lem.fy}.
 Otherwise, $\hat{K}_2 \in \Afam \setminus \{\hat{X}\}$.
 Even in this case, there is a contradiction because each core in $\Afam
 \setminus \{\hat{X}\}$ is covered by no edge in $F$.

 Suppose that $\hat{K}_1$ is included by $\hat{X}$.
 $\hat{X} \cup \hat{Z} \in \Cfam(\hat{M})$ holds.
 Moreover, $\hat{X}\subset \hat{X} \cup \hat{Z}$ holds
 because $Z^+$ includes $h$, 
 and $\hat{Z}\subset \hat{X} \cup \hat{Z}$ holds
 because $\{a,b\} \in \delta_E(\hat{Z})$ is included by $X^+$.
 $\hat{X}\cup \hat{Z}$ is covered by some edge $f' \in F$.
 The witness of $f'$ is incomparable with $\hat{X}$
 since otherwise, $f' \subseteq X^+$.
 $f'$ covers $\hat{X}$ or $\hat{Z}$.
 If $f'$ covers $\hat{Z}$,
 then $f'$ is chosen instead of $e$, and this case is categorized into
 the previous one where $\hat{K}_1$ is incomparable with $\hat{X}$.
 Hence, $f'$ covers $\hat{X}$. Then,
 Lemmas~\ref{lem.uncrossable}~(iii) and \ref{lem.afam} indicate that
 all cores comparable with the witness of
 $f'$ 
 are covered by $F$ before $\hat{X}$ is added to $\Afam$,
 which is
 a contradiction.
\end{proof}

 \begin{lemma}\label{lem.spider2}
 When $|B|\geq 2$,
  the edge set $S$
  is a spider with $|B|$ feet, and $h$ is its head.
\end{lemma}
\begin{proof}
 Let $\Bfam=\{\hat{B}_1,\ldots,\hat{B}_b\}$ be the set of witnesses of
 the edges in $B$.
 Let $\hat{M}_{b'}$ be the min-core included by $\hat{B}_{b'}$,
 and let $F_{b'}$ denote $F_{B_{b'}}$
 for each $b' \in [b]$.
 Lemma~\ref{lem.fy} shows that $F_{b'}$ covers $\Cfam(\hat{M}_{b'},h)$
 for each $b' \in [b]$.
 Hence it suffices to prove that 
 $V(F_{b_1}) \cap V(F_{b_2}) \subseteq \{h\}$ for each
 $b_1,b_2\in [b]$ with $b_1 \neq b_2$.
 Suppose that $e_1 \in F_{b_1}$ and $e_2 \in F_{b_2}$ share an
 end-node $v$ with $h\neq v$.
 
 Suppose that $e_1$ was added to $F$
 before $e_2$.
 Let $\hat{Y}_1$
 be the witness of $e_1$, and $\hat{Y}'_1$ be the core 
 that was added to $\Afam$ when $\hat{Y}_1$ was removed from $\Afam$.
 Note that $\hat{Y}_1 \subset \hat{Y}'_1$, and
 $e_1$ does not cover $\hat{Y}'_1$ but $\hat{Y}_1$.
 Hence, $v \in (Y'_1)^+$,
 and the other end-node of $e_2$ is in $B_{b_2}$.
 If $v \in Y'_1$, then
 $e_2$ covers all cores including $B_{b_2}$
 since they are strongly disjoint with $Y'_1$.
 Hence, Case (b) occurred when $e_2$ was added to $F$, 
 and $v$ must be $h$ in this case.
 Even if $v \not\in Y'_1$, 
 $e_1$ and $e_2$ are added to $F$ because of the directed edges leaving $v$.
 This means that
 Case (b) occurred when $e_2$ was added to $F$,
 and $h=v$ holds.
\end{proof}

\begin{lemma}\label{lem.spider-lp}
 There exists $w\colon V \rightarrow W$ such that $S$ is
 activated by $w$, and $w(V)/f(S) \leq \sum_{\hat{X} \in
 \Lfam}z(\hat{X})/|\Mfam_{\Vfam}|$.
\end{lemma}
\begin{proof}
 Recall that each edge in $S$ is undirected, but it has a unique
 direction in which it enters the inner-part of its witness.
 Hence, we regard the edges in $S$ as directed edges in this proof.
 For each $e=uv \in S$,
 there exists $(j_e,j'_e)\in \Psi^{e}$
 such that 
 \eqref{dual-c3'} is tight for 
 $(u,j_e)$ and $\eqref{dual-c3}$ is tight for $(v,j'_e)$.
 We can activate $e$ by setting $w(u)$ to a value of at least $j_e$ and
 $w(v)$ to a value of at least $j'_e$.
 When $e$ is added to $F$, $e$ assigns $j_e$ to $u$ and $j'_e$ to $v$.
 If a node has incident edges in $S$,
 we set the weight of the node to the maximum value assigned from the
 incident edges in $S$.
 If a node has no incident edge in $S$, then its weight is set to $0$.
 Let $\tau$ be the time when the algorithm was completed.
 Below, we prove that the total weight assigned from edges in $S$
 is at most $\tau f(S)$ where we do not count
 the weight assigned to the head $h$ of $S$ multiple times.
  Since $\tau = \sum_{\hat{X}\in \Lfam}z(\hat{X})/|\Mfam|$, this proves
 the lemma.

 Let $\hat{M}$ be a foot of $S$ and
 $S'$ be the set of edges in $S$ that cover $\Cfam(\hat{M},h)$.
 Let $e=uv \in S'$.  $e$ assigns $j_e \in W$
 to $u$  and $j'_e$ to $v$.
 Moreover,
\begin{equation}
 j_e=\sum_{\hat{X} \in \Lfam: u\in V \setminus X^+}z(\hat{X},u,j_e)
\label{eq.je}
\end{equation}
 holds because \eqref{dual-c3'} is tight for $(u,j_e)$,
 and 
\begin{equation}
 j'_e=\sum_{\hat{X} \in \Lfam: v\in X}z(\hat{X},v,j'_e) 
\label{eq.je'}
\end{equation}
 holds because \eqref{dual-c3} is tight for $(v,j'_e)$. 
 Let $\tau_e$ denote the time when  \eqref{dual-c3'} became tight for $(u,j_e)$.

 We first consider the case where $u \neq h$.
 Let us prove that the right-hand side of \eqref{eq.je} is contributed
 by cores covered by $e$. 
 Suppose that $z(\hat{X},u,j_{e})>0$ holds for some $\hat{X} \in \Lfam$ with
 $u \in V \setminus X^+$. 
 Then there exists an edge
 $uv'$ that covers $\hat{X}$, and
 \eqref{dual-c3} was tight for some $(v',j')$
 with $(j_e,j') \in \Psi^{uv'}$
 at time $\tau_e$.
  If $\hat{X} \not\in \Lfam(\hat{M})$,
 then
 this means that Case (b) occurred when $e$ was added to $F$.
 Since this contradicts $u\neq h$,
 we have $\hat{X} \in \Lfam(\hat{M})$.
 If $\hat{X}$ includes the witness of $e$,
 then $e$ covers $\hat{X}$ because $u \not\in X^+$.
 Hence,
 $\hat{X}$ is included by the
 witness of $e$. However, in this case, 
 $uv$ is not added to $F$ by the algorithm.
 Hence $e$ covers $\hat{X}$.

 The right-hand side of \eqref{eq.je'} is also contributed by
 cores covered by $e$.
 To see this, suppose that $z(\hat{X},v,j'_e)>0$
 holds for some $\hat{X} \in \Lfam$ with $v \in X$.
 If $e$ does not cover $\hat{X}$,
 then $u \in X^+$ holds, implying that $e$ was already in $F$ when $\hat{X}$ entered
 $\Afam$.
 In other words, $\hat{X}$ enters $\Afam$ after time $\tau_e$.
 However, \eqref{dual-c3} was tight for $(v,j'_e)$ at time $\tau_e$.
 Therefore, $z(\hat{X},v,j'_e)>0$ does not hold unless $e$ covers
 $\hat{X}$. Note that this is the case even when $u=h$.

 When $u=h$, $e$ assigns $j_e$ to $h$ but 
 more than one edge leaving $h$ in $S$ may assign the same weight to $h$.
 By the same discussion as above, 
 if a core $\hat{X} \in \Lfam$ with $h\not\in X^+$ satisfies
 $z(\hat{X},h,j_e)>0$,
 then $S$ contains an edge that leaves $h$ and covers $\hat{X}$.
 Hence, we here count only $\sum_{\hat{X}\in \Lfam(\hat{M},h)}z(\hat{X},h,j_e)$ as the weight assigned from $e$ to $h$.
 A core $\hat{X} \in \Lfam(\hat{M},h)$ contributing to this value is covered by
 $e$ according to the discussion above.
 Then the total weight assigned from edges in $S'$ is exactly
 \[
 \sum_{e=uv \in S'}  \sum_{\hat{X}\in \Lfam(\hat{M}): e \in \delta^-_A(\hat{X})}
 \left(
 z(\hat{X},u,j_e)+ z(\hat{X},v,j'_e)
\right)
= \sum_{e \in S'}  \sum_{\hat{X}\in \Lfam(\hat{M}): e \in \delta^-_A(\hat{X})}z(\hat{X}).
 \]
 Lemma~\ref{lem.fy} tells that each $\hat{X} \in \Lfam$ is covered by exactly one edge in $S$.
 Hence the right-hand side of the above equality is equal to $\sum_{\hat{X} \in \Lfam(\hat{M})}z(\hat{X})$.
 Since two cores in $\Lfam(\hat{M})$ do not belong to $\Afam$ simultaneously, 
 this does not exceed $\tau$.
 Since $S$ has $f(S)$ feet, it implies that 
 the total weight is at most $\tau f(S)$.
 \end{proof}

 Theorem~\ref{thm.spideralgorithm} follows from
 Lemmas~\ref{lem.spider1}, \ref{lem.spider2},
 and \ref{lem.spider-lp}.

\section{Potential function on uncrossable biset families}
\label{sec.potential}

In this section, $\Vfam$ is an uncrossable family of bisets
and $\gamma$ stands for $\max_{\hat{X} \in \Vfam}|\Gamma(\hat{X})|$.

For analyzing the greedy algorithm of choosing spiders repeatedly, we need a potential function that
measures the progress of the algorithm. Nutov~\cite{Nutov12uncrossable} used $|\Mfam_{\Vfam}|$ as a
potential. He claimed that this potential gives $O(\log d)$-approximation  because
$|\Mfam_{\Vfam}|-|\Mfam_{\Vfam_S}| \geq f(S)/3$ holds for each uncrossable biset family $\Vfam$ and
each spider $S$ of $\Vfam$. However, there is a case with $|\Mfam_{\Vfam}|-|\Mfam_{\Vfam_S}|=0$ as
follows. Let $\Vfam=\{\hat{X}_1,\hat{Y}_1,\ldots,\hat{X}_n,\hat{Y}_n\}$, and suppose that 
$\hat{X}_l \subseteq \hat{Y}_l$ for each $l \in [n]$, $\hat{Y}_l$ and $\hat{Y}_{l'}$ are strongly disjoint for
each $l,l' \in [n]$ with $l\neq l'$, and a node $h$ is in $\Gamma(\hat{Y}_l) \setminus X^+_l$ for
each $l \in [n]$. $\Vfam$ is strongly laminar, and hence uncrossable. Note that
$\Mfam_{\Vfam}=\{\hat{X}_1,\ldots,\hat{X}_n\}$, and hence $|\Mfam_{\Vfam}|=n$. If the head of a
spider $S$ is $h$ and its feet are $\hat{X}_1,\ldots,\hat{X}_n$ (i.e., $f(S)=n$), then
$\Mfam_{\Vfam_S}=\{\hat{Y}_1,\ldots,\hat{Y}_n\}$ holds, and hence $|\Mfam_{\Vfam_S}|=n$. Therefore,
$|\Mfam_{\Vfam}|-|\Mfam_{\Vfam_S}|=0$.

Vakilian~\cite{Vakilian13} showed that such an inconvenient situation does not appear if $\Vfam$ arises from the
node-weighted SNDP. To explain this more precisely, let $(V,E_0)$ be the graph to be augmented in an
instance of the prize-collecting augmentation problem. Recall that the problem requires to add edges in an edge set $E$ to $E_0$. If this instance is obtained by the reduction from the
node-weighted SNDP in Theorem~\ref{thm.augmentation}, then $E_0$ is the subset of $E_0\cup E$ induced by some node set 
$U\subseteq V$, and each biset $\hat{X}$ that requires to be covered satisfies $\Gamma(\hat{X})\subseteq U$.
Moreover, a spider is not chosen if its head is in $U$, 
and therefore the heads of chosen spiders are not included by the neighbor
of any biset. This means that each spider $S$ achieves $|\Mfam_{\Vfam}|-|\Mfam_{\Vfam_S}| \geq f(S)/3$ for $\Vfam$ arising
from the node-weighted SNDP. However this is not the case for all uncrossable biset families, including those arising from 
the PCNAP because $(V,E_0)$ may not be an induced subgraph in general.

Because of this, using $|\Mfam_{\Vfam}|$ as a potential function gives no desired
approximation guarantee for general uncrossable biset families. Hence, we introduce a new potential function in this section. For a family
$\mathcal{X}$ of cores and core $\hat{X} \in \mathcal{X}$,  let $\Delta_{\mathcal{X}}(\hat{X})$
denote the set of nodes $v\in \Gamma(\hat{X})$ such that  there exists another core $\hat{Y} \in
\mathcal{X} \setminus \{\hat{X}\}$ with $v \in \Gamma(\hat{Y})$. We define the potential
$\phi_{\mathcal{X}}(\hat{X})$ of a core $\hat{X}$ as $\gamma-|\Delta_{\mathcal{X}}(\hat{X})|$. The
potential $\phi(\mathcal{X})$ of $\mathcal{X}$ is defined as 
$(\gamma+1)|\mathcal{X}|+\sum_{\hat{X} \in \mathcal{X}}\phi_{\mathcal{X}}(\hat{X})$.

\begin{lemma}\label{lem.neighbor}
 Let $\hat{X} \in \Mfam_{\Vfam}$, $S$ be an edge set, and $\hat{Y}$ be the
 min-core in $\Mfam_{\Vfam_S}$ 
 such that $\hat{X}
 \subseteq \hat{Y}$ where $\hat{X}=\hat{Y}$ possibly holds.
 Let $v$ be a node with
 $v \in \Delta_{\Mfam_{\Vfam}}(\hat{X}) \setminus \Delta_{\Mfam_{\Vfam_S}}(\hat{Y})$,
 and $\hat{Z}$ be a min-core in $\Mfam_{\Vfam} \setminus \{\hat{X}\}$ with $v \in \Gamma(\hat{Z})$.
 Then, $S$ covers all cores in $\Cfam_{\Vfam}(\hat{Z})$. If there exists a min-core in
 $\Mfam_{\Vfam_S}$ that includes $\hat{Z}$, then it is $\hat{Y}$.
\end{lemma}
\begin{proof}
 Since $v \in \Gamma(\hat{X}) \subseteq Y^+$,
 $v$ is either in $Y$ or $\Gamma(\hat{Y})$.
 Suppose it is the former case (i.e., $v \in Y$). Then, $\hat{Z} \not\in
 \Vfam_S$ 
 because $\hat{Y}$ and $\hat{Z}$ are not strongly disjoint in this case,
 and $\hat{Z} \in \Vfam_S$
 contradicts Lemma~\ref{lem.uncrossable}~(iii).
 Moreover, $\hat{Z}$ is
 included by $\hat{Y}$ 
 since, otherwise, they must be strongly disjoint, contradicting the
 existence of $v$.
 This means that all cores
 in $\Cfam_{\Vfam}(\hat{Z})$ are covered by $S$.

 Suppose it is the latter case (i.e., $v \in \Gamma(\hat{Y})$). 
 Let $\hat{Z}'$ be a min-core in $\Mfam_{\Vfam_S}$ that includes $\hat{Z}$, and
 assume that it is distinct from $\hat{Y}$.
 Since $v \not\in \Delta_{\Mfam_{\Vfam_S}}(\hat{Y})$, no min-core
 in $\Mfam_{\Vfam_S} \setminus \{\hat{Y}\}$ contains $v$ in its neighbor.
 Hence $v \in Z'$.
 However, this means that 
 $\hat{Z}'$ and $\hat{Y}$ are not strongly disjoint, which contradicts
 Lemma~\ref{lem.uncrossable}~(iii).
 This implies that $S$ covers $\Cfam_{\Vfam}(\hat{Z})$ since,
 if $\Cfam_{\Vfam}(\hat{Z})$ contains a core not covered by $S$, then
 the minimal core among such cores is a min-core in $\Mfam_{\Vfam_S}$ distinct
 from $\hat{Y}$.
\end{proof}

\begin{lemma}\label{lem.newcore}
 Let $S$ be an edge set and $\hat{Y} \in \Mfam_{\Vfam_S} \setminus \Mfam_{\Vfam}$.
 Then, exactly one of the following holds\/{\rm :}
\begin{itemize}
 \item $\hat{Y}$ includes at least two min-cores in $\Mfam_{\Vfam}
       \setminus \Mfam_{\Vfam_S}$,
       and all cores of $\Vfam$ including these min-cores are covered by
       $S$\/.
 \item $\hat{Y}$ is a core of $\Vfam$ that includes 
       a min-core in $\Mfam_{\Vfam} \setminus \Mfam_{\Vfam_S}$.
\end{itemize}
\end{lemma}
\begin{proof}
 Since $\hat{Y} \not\in \Mfam_{\Vfam}$, there exist min-cores in $\Mfam_{\Vfam}$
 included by $\hat{Y}$. Suppose that the number of such min-cores is
 one, and we call the min-core by $\hat{X}$. Then,
 $\hat{Y}$ is a core of $\Vfam$.
 Since $\hat{Y} \in \Mfam_{\Vfam_S}$, $\hat{X}$ is covered by $S$, and hence,
 $\hat{X} \in \Mfam_{\Vfam} \setminus \Mfam_{\Vfam_S}$.
 If the number of such min-cores is at least two, 
 then the cores of $\Vfam$ including such min-cores are covered by $S$
 because $\hat{Y}$ is minimal in $\Vfam_S$.
\end{proof}

\begin{lemma}\label{lem.potentialdecrease}
 Let $S$ be a spider for $\Vfam$. If $f(S)=1$, then 
 $\phi(\Mfam_{\Vfam}) - \phi(\Mfam_{\Vfam_S}) \geq 1$.
 Otherwise, $\phi(\Mfam_{\Vfam}) - \phi(\Mfam_{\Vfam_S}) \geq (f(S)-1)/2$.
\end{lemma}
\begin{proof}
 Let $\nu(S)$ denote the number of min-cores
 $\hat{X} \in \Mfam_{\Vfam}$
 such that $S$ covers all bisets in $\Cfam_{\Vfam}(\hat{X})$, and let $\xi(S)$ denote the
 number of min-cores $\hat{Y} \in \Mfam_{\Vfam}$ 
 such that $S$ covers $\hat{Y}$ but not all bisets in $\Cfam_{\Vfam}(\hat{Y})$.
 Note that $\nu(S)+\xi(S)\geq f(S)$ holds.
 If $\hat{Y}$ is a min-core counted in $\xi(S)$,
 then there exists a unique min-core $\hat{Y}' \in \Mfam_{\Vfam_S}$ that
 includes $\hat{Y}$. Let $\Pfam$ denote the set of pairs of such $\hat{Y}$ and $\hat{Y}'$.

 Let $\hat{X} \in \Mfam_{\Vfam}$ be a min-core counted in $\nu(S)$.
 If a core of $\Vfam_S$ includes $\hat{X}$, then
 the core includes at least two min-cores in $\Mfam_{\Vfam}$.
 Let $\Mfam_1$ be the set of such $\hat{X}$ that
 is included by a min-core in $\Vfam_S$, 
 and let $\Mfam_2$ be the set of such $\hat{X}$
 that 
 is included by no min-core of $\Vfam_S$ (although it may be included by a core in $\Vfam_S$).
 Note that $|\Mfam_1|+|\Mfam_2|= \nu(S)$.

 By Lemma~\ref{lem.newcore},
 each min-core in $\Mfam_{\Vfam_S} \setminus \Mfam_{\Vfam}$  
 includes at least two members of $\Mfam_1$ 
 or belongs to $\Cfam_{\Vfam}(\hat{Y})$ defined by a min-core 
 $\hat{Y}\in \Mfam_{\Vfam}$ covered by $S$.
 Hence
 $|\Mfam_{\Vfam_S} \setminus \Mfam_{\Vfam}| \leq |\Mfam_1| /2 +  \xi(S)$.
 From this, it follows that 
 \[
  |\Mfam_{\Vfam_S}| \leq |\Mfam_{\Vfam_S} \setminus \Mfam_{\Vfam}| + |\Mfam_{\Vfam}| -\nu(S) - \xi(S)
 \leq |\Mfam_{\Vfam}| - \frac{|\Mfam_1|}{2} - |\Mfam_2|.
\]
 
 Recall that
 $\phi(\Mfam_{\Vfam})$ is defined as $(\gamma+1)|\Mfam_{\Vfam}|+\sum_{\hat{Z} \in
 \Mfam_{\Vfam}}\phi_{\Mfam_{\Vfam}}(\hat{Z})$,
 and 
 $\phi(\Mfam_{\Vfam_S})$ is defined as 
 $(\gamma+1)|\Mfam_{\Vfam_S}|+\sum_{\hat{Z} \in \Mfam_{\Vfam_S}}\phi_{\Mfam_{\Vfam_S}}(\hat{Z})$.
 The first term of $\phi(\Mfam_{\Vfam})$
 is larger than that of $\phi(\Mfam_{\Vfam_S})$ by $(\gamma+1)(|\Mfam_{\Vfam}|-|\Mfam_{\Vfam_S}|)$.
 A min-core $\hat{Z} \in \Mfam_{\Vfam_S}\setminus \Mfam_{\Vfam}$ either includes at least two members of
 $\Mfam_1$ or belongs to $\Cfam_{\Vfam}(\hat{Y})$ defined by a
 min-core $\hat{Y} \in \Mfam_{\Vfam}\setminus \Mfam_{\Vfam_S}$ (i.e., $(\hat{Y},\hat{Z})
 \in \Pfam$).
 There are at most $|\Mfam_1|/2$ min-cores of the former type, and hence
 the sum of their potentials is at most $\gamma |\Mfam_1|/2$.
 Let $\hat{Z}$ belong to the latter type. 
 Note that
 \[
  \phi_{\Mfam_{\Vfam}}(\hat{Y})-\phi_{\Mfam_{\Vfam_S}}(\hat{Z})
 = |\Delta_{\Mfam_{\Vfam_S}}(\hat{Z})|-|\Delta_{\Mfam_{\Vfam}}(\hat{Y})|
 =
 |\Delta_{\Mfam_{\Vfam_S}}(\hat{Z})\setminus \Delta_{\Mfam_{\Vfam}}(\hat{Y})|-
 |\Delta_{\Mfam_{\Vfam}}(\hat{Y})\setminus
 \Delta_{\Mfam_{\Vfam_S}}(\hat{Z})|.
\]
 If there exists $v \in \Delta_{\Mfam_{\Vfam}}(\hat{Y})\setminus \Delta_{\Mfam_{\Vfam_S}}(\hat{Z})$,
 then there exists $\hat{C} \in \Mfam_{\Vfam}$ counted in $\nu(S)$ such that $v
 \in \Gamma(\hat{C})$, and $\hat{C} \in \Mfam_2$ by Lemma~\ref{lem.neighbor}. 
 We make $\hat{C}$ give one token to $\hat{Z}$.
 Then, $\hat{Z}$ obtains $|\Delta_{\Mfam_{\Vfam}}(\hat{Y})\setminus
 \Delta_{\Mfam_{\Vfam_S}}(\hat{Z})|$ tokens. 
 Note that only $\hat{Z}$ contains $v$ in its
 outer-part among all min-cores in $\Mfam_{\Vfam_S}$;
 If $v \in Z$, then
 it is implied by the strong disjointness of min-cores, and 
 if $v \in \Gamma(\hat{Z})$, then it is implied by $v \not\in
 \Delta_{\Mfam_{\Vfam_S}}(\hat{Z})$.
 Hence, each $\hat{C} \in \Mfam_2$ releases at most
 one token for each node $v \in \Gamma(\hat{C})$.
Therefore, the total number of tokens is at most $\gamma|\Mfam_2|$, and
hence,
\[
 \sum_{(\hat{Y},\hat{Z}) \in \Pfam}
 |\Delta_{\Mfam_{\Vfam}}(\hat{Y})\setminus \Delta_{\Mfam_{\Vfam_S}}(\hat{Z})|
 \leq \gamma |\Mfam_2|.
\]
 Summing up, 
\begin{align}
 &\phi(\Mfam_{\Vfam})-\phi(\Mfam_{\Vfam_S})\notag\\
 &\geq  (\gamma+1)(|\Mfam_{\Vfam}|-|\Mfam_{\Vfam_S}|) - \frac{\gamma |\Mfam_1|}{2}
  +
  \sum_{(\hat{Y},\hat{Z}) \in \Pfam}(
 |\Delta_{\Mfam_{\Vfam_S}}(\hat{Z})\setminus \Delta_{\Mfam_{\Vfam}}(\hat{Y})|-
 |\Delta_{\Mfam_{\Vfam}}(\hat{Y})\setminus \Delta_{\Mfam_{\Vfam_S}}(\hat{Z})|)
\notag \\
 &\geq 
 (\gamma+1)\left(\frac{|\Mfam_1|}{2}+|\Mfam_2|\right) - \frac{\gamma|\Mfam_1|}{2} - \gamma|\Mfam_2| + 
  \sum_{(\hat{Y},\hat{Z}) \in \Pfam} |\Delta_{\Mfam_{\Vfam_S}}(\hat{Z})\setminus \Delta_{\Mfam_{\Vfam}}(\hat{Y})|
\notag\\
&=\frac{|\Mfam_1|}{2} + |\Mfam_2| + 
  \sum_{(\hat{Y},\hat{Z}) \in \Pfam} |\Delta_{\Mfam_{\Vfam_S}}(\hat{Z})\setminus \Delta_{\Mfam_{\Vfam}}(\hat{Y})|\notag\\
&\geq  \frac{\nu(S)}{2} + 
  \sum_{(\hat{Y},\hat{Z}) \in \Pfam} |\Delta_{\Mfam_{\Vfam_S}}(\hat{Z})\setminus \Delta_{\Mfam_{\Vfam}}(\hat{Y})|.\label{eq.potential-lowerbound}
\end{align}

 If $f(S)=1$, then $\nu(S) \geq 1$, and hence $\phi(\Mfam_{\Vfam})-\phi(\Mfam_{\Vfam_S})
 \geq 1/2$ by \eqref{eq.potential-lowerbound}. Since potentials are integers, this means that 
 $\phi(\Mfam_{\Vfam})-\phi(\Mfam_{\Vfam_S}) \geq 1$.
 Suppose that $f(S) \geq 2$. 
  Consider the case where the head of $S$ is included by the inner-part
 of some min-core $\hat{X} \in \Mfam_{\Vfam_S}$.
 If a foot $\hat{C}$ of $S$ is strongly disjoint from $\hat{X}$, 
 then $\Cfam_{\Vfam}(\hat{C})$ is covered by $S$, and hence $\hat{C}$ is counted in $\nu(S)$.
 If $\hat{X}$ includes at least two feet of $S$, then all cores of $\Vfam$
 including these feet are covered by $S$. Therefore,
 $\nu(S) \geq f(S)-1$, and hence
 $\phi(\Mfam_{\Vfam})-\phi(\Mfam_{\Vfam_S})\geq (f(S)-1)/2$ by \eqref{eq.potential-lowerbound}.

 In the remaining case, $f(S) \geq 2$ and 
 no min-core in $\Mfam_{\Vfam_S}$ contains
the head $h$ of $S$ in its inner-part.
By definition of spiders, each foot $\hat{C}$ is covered by $S$.
 Hence $\hat{C}$ is counted in $\nu(S)$ or $\xi(S)$.
 If $\nu(S) \geq f(S)-1$, then we are done. Hence, suppose that
 $\nu(S)\leq f(S)-2$. $f(S)-\nu(S)$ feet of $S$ are counted in $\xi(S)$.
 Let $\hat{Y}$ be a foot of $S$ that is counted in $\xi(S)$.
 Then, there exists $\hat{Z} \in \Mfam_{\Vfam_S}$ with
 $(\hat{Y},\hat{Z})\in \Pfam$ and $h \in \Gamma(\hat{Z}) \setminus
 \Gamma(\hat{Y})$.
 Since $\Mfam_{\Vfam_S}$ contains at least two such $\hat{Z}$,
 we have $h \in \Delta_{\Mfam_{\Vfam_S}}(\hat{Z}) \setminus
 \Delta_{\Mfam_{\Vfam}}(\hat{Y})$.
 Therefore,
\[
 \nu(S) + 
  \sum_{(\hat{Y},\hat{Z}) \in \Pfam} 
 |\Delta_{\Mfam_{\Vfam_S}}(\hat{Z})\setminus \Delta_{\Mfam_{\Vfam}}(\hat{Y})|
 \geq f(S),
\]
 and \eqref{eq.potential-lowerbound} implies that
 $\phi(\Mfam_{\Vfam})-\phi(\Mfam_{\Vfam_S}) \geq f(S)/2$.
 \end{proof}

\begin{theorem}\label{thm.potential-spider}
 Let $\Vfam$ be an uncrossable family of bisets.
 There exist $w\colon V \rightarrow W$,
 a spider $S$ activated by $w$,
 and a strongly laminar family $\Lfam$ of cores of $\Vfam$
 such that 
 \[
  \frac{w(V)}{\phi(\Mfam_{\Vfam})-\phi(\Mfam_{\Vfam_{S}})} = O(\max\{\gamma,1\})
 \cdot \frac{\CoreLP(\Lfam)}{\phi(\Mfam_{\Vfam})}.
 \]
\end{theorem}
\begin{proof}
 Theorem~\ref{thm.spideralgorithm} shows that
 there exist $w\colon V\rightarrow W$, 
 a spider $S$ activated by $w$, and a strongly laminar family $\Lfam$ of cores
 such that
 \[
  \frac{w(V)}{f(S)} \leq \frac{\CoreLP(\Lfam)}{|\Mfam_{\Vfam}|}.
 \]
 Since $\phi(\Mfam_{\Vfam}) \leq (2\gamma+1)|\Mfam_{\Vfam}|$, we have
\begin{equation}\label{eq.potential-1}
 \frac{w(V)}{f(S)} \leq  \frac{\CoreLP(\Lfam)}{|\Mfam_{\Vfam}|} \leq
 (2\gamma+1) \cdot \frac{\CoreLP(\Lfam)}{\phi(\Mfam_{\Vfam})}.
\end{equation}

 If $f(S)=1$, then $\phi(\Mfam_{\Vfam})-\phi(\Mfam_{\Vfam_S}) \geq f(S)$ by
 Lemma~\ref{lem.potentialdecrease}, and hence, the required inequality
 follows from \eqref{eq.potential-1}.
 Otherwise,
 $\phi(\Mfam_{\Vfam})-\phi(\Mfam_{\Vfam_S}) \geq (f(S)-1)/2$ by
 Lemma~\ref{lem.potentialdecrease}, and hence,
\[
  \frac{w(V)}{f(S)} \geq 
  \frac{w(V)}{2(f(S)-1)} \geq 
\frac{w(V)}{4(\phi(\Mfam_{\Vfam})-\phi(\Mfam_{\Vfam_S}))},
\]
 where the first inequality follows from $f(S)\geq 2$.
Combining with \eqref{eq.potential-1}, this gives
\[
\frac{w(V)}{\phi(\Mfam_{\Vfam})-\phi(\Mfam_{\Vfam_S})} \leq
 4(2\gamma+1)\cdot \frac{\CoreLP(\Lfam)}{\phi(\Mfam_{\Vfam})}.
\]
\end{proof}

Our algorithm presented in Section~\ref{sec.primal-dual}
computes the node weights $w$ and spider $S$ claimed by
Theorem~\ref{thm.potential-spider} in polynomial time. 
Alternatively, one can use the 
simpler algorithm in~\cite{Nutov13activation}, which approximates $w$
within a factor of $2$.

\section{Algorithm}
\label{sec.algorithm}

We first present our main theorem.

\begin{theorem}\label{thm.main}
 Suppose that $\Vfam$ is a biset family such that $\bigcup_{i \in D}\Vfam_i$ is
 uncrossable for each $D \subseteq [d]$.
 Let $\gamma=\max_{\hat{X}\in \Vfam}|\Gamma(\hat{X})|$
 and $\gamma'=\max\{\gamma,1\}$.
 The prize-collecting biset covering problem with $\Vfam$ admits an $O(\gamma'
 \log(\gamma' d))$-approximation algorithm.
\end{theorem}
\begin{proof}
 Let $(x,y)$ be an optimal solution for $\LP(\Vfam)$.
 We first compute $(x,y)$.
 We eliminate all demand pairs $\{s_i,t_i\}$ such that $y(i)\geq 1/2$,
 and eliminate each biset that separates no remaining demand pair from $\Vfam$.
 Let $\Vfam'$ be the biset family obtained after this operations.
 $\NPCLP(\Vfam') \leq 2 \sum_{v \in V}\sum_{j \in W}j \cdot x(v,j)$ holds
 because $2x$ is feasible to $\NPCLP(\Vfam')$.

 Applying Theorem~\ref{thm.potential-spider} to $\Vfam'$,
 we obtain $w$, $S$ and $\Lfam$ 
 such that
 $w(V)/(\phi(\Mfam_{\Vfam'})-\phi(\Mfam_{\Vfam'_S}))=O(\gamma')\cdot
 \CoreLP(\Lfam)/\phi(\Mfam_{\Vfam'})$,
and the right-hand side is at most $O(\gamma') \cdot \NPCLP(\Vfam')/\phi(\Mfam_{\Vfam'})$
 by Lemma~\ref{lem.corelpvslp}.
 If $\phi(\Mfam_{\Vfam'_S})> 0$, then
 we apply
 Theorem~\ref{thm.potential-spider} to $\Vfam'_S$.
 Let $w'$ and $S'$ be the obtained node weights and spider, respectively.
 We add edges in $S'$ to $S$,
 increase the weight $w(v)$ by $w'(v)$ for each $v\in V$.
 We repeat this until $\phi(\Mfam_{\Vfam'_S})$ becomes 0.
 By a standard argument of the greedy algorithm for the set cover problem,
 we have $w(V) = O(\gamma' \log(\phi(\Mfam_{\Vfam'}))) \cdot \NPCLP(\Vfam')$
 when the above procedure is completed.
 Since $\phi(\Mfam_{\Vfam'})=O(\gamma' d)$,
 it implies that $w(V)=O(\gamma' \log(\gamma' d)) \cdot \NPCLP(\Vfam')$.

The penalty of $w$ is at most $2\sum_{i \in [d]}\pi_i y(i)$
 because $S$ covers all bisets separating each demand pair $\{s_i,t_i\}$
 with $y(i)<1/2$, and $S\subseteq E_w$.
 $w(V) = O(\gamma' \log(\gamma' d)) \cdot \NPCLP(\Vfam')
 =  O(\gamma' \log(\gamma' d)) \cdot\sum_{j\in W}\sum_{v \in V}j\cdot x(v,j)$.
 Therefore the objective value of $w$ is 
$O(\gamma' \log(\gamma' d))$ times $\LP(\Vfam)$.
 Lemma~\ref{lem.relaxation} shows that $\LP(\Vfam)$ is at most the
 optimal value of the prize-collecting biset covering problem.
\end{proof}

\begin{corollary}\label{cor.edge-element}
 Let $k'=\min\{k,|V|\}$.
 The edge-connectivity PCNAP admits an $O(k \log d)$-approximation
 algorithm, and
 the element-connectivity PCNAP admits an $O(k k' \log(k'd))$-approximation
 algorithm.
 \end{corollary}
\begin{proof}
$\bigcup_{i \in [d]}\Vfam_i^{\rm edge}$ is an uncrossable family of bisets with $\gamma=0$.
 Hence, Theorems~\ref{thm.augmentation} and \ref{thm.main} give
 an $O(k \log d)$-approximation algorithm for the edge-connectivity PCNAP.
$\bigcup_{i \in [d]}\Vfam^{\rm ele}$ is an uncrossable family of bisets
 with $\gamma\leq k'-1$.
 Hence, Theorems~\ref{thm.augmentation} and \ref{thm.main} give
 an $O(kk' \log (k'd))$-approximation algorithm for the element-connectivity PCNAP.
\end{proof}

We note that 
$d=O(|V|^2)$. Hence, the above corollary gives
an $O(k\log |V|)$-approximation algorithm for the edge-connectivity
PCNAP,
and an $O(k^2\log |V|)$-approximation algorithm for the element-connectivity
PCNAP.

The next corollary provides approximation algorithms for the
node-connectivity requirements.
Since it is reasonable to suppose $k\leq |V|$ for the node-connectivity
requirements,
the next corollary does not have $k'$ in contrast with Corollary~\ref{cor.edge-element}.

\begin{corollary}\label{cor.node}
\begin{itemize}
 \item[\rm (i)] The node-connectivity PCNAP admits an
	      $O(k^5 \log |V|\log(kd))$-approximation randomized algorithm.
 \item[\rm (ii)] The rooted node-connectivity PCNAP admits an
	      $O(k^3 \log(kd))$-approximation algorithm.
 \item[\rm (iii)] The subset node-connectivity PCNAP admits an
	      $O(k^3 \log(kd))$-approximation algorithm.
\end{itemize}
 \end{corollary}
\begin{proof}
 Theorem~\ref{thm.augmentation}
 reduces the node-connectivity PCNAP
 to the prize-collecting biset covering problem with the biset family
 $\Vfam=\bigcup_{i\in [d]}\Vfam^{\rm node}_i$ by paying factor $k$.
 Chuzhoy and Khanna~\cite{ChuzhoyK12} presented a randomized algorithm
 for decomposing an instance of the node-connectivity SNDP into 
 $O(k^3 \log |V|)$ instances of the element-connectivity SNDP such that
 the union of solutions for the $O(k^3 \log |V|)$ instances is 
 feasible to the original instance.
 This algorithm can be applied for
 computing $O(k^3 \log |V|)$  
 uncrossable subfamilies of $\Vfam$
 such that an edge set covering the union of the subfamilies
 covers $\Vfam$.
 By Theorem~\ref{thm.main}, 
 we compute $O(k \log (kd))$-approximate solutions for
 instances of 
 the prize-collecting biset covering problem 
 with the subfamilies. We then return the union of the obtained solutions.
 This achieves
 $O(k^5 \log (kd) \log |V|)$-approximation for the original instance
 of the node-connectivity PCNAP.

 For the rooted node-connectivity PCNAP, 
 we replace the decomposition result due to Chuzhoy and Khanna~\cite{ChuzhoyK12} by
 the one due to Nutov~\cite{Nutov12uncrossable}, which proved
 that $\Vfam$
 can be decomposed into $O(k)$ uncrossable subfamilies.
 This achieves $O(k^3 \log (kd))$-approximation for the rooted
 node-connectivity PCNAP.

 Strictly speaking,
 Theorem~\ref{thm.augmentation} cannot be applied to the subset
 node-connectivity PCNAP because it is not a special case of the PCNAP, 
 but we can similarly prove that the same claim holds for
 the subset node-connectivity PCNAP.
  Using a decomposition result in Nutov~\cite{Nutov12subset},
 the augmentation problem obtained by the reduction 
 can be decomposed into
 one instance with the rooted node-connectivity requirements and $O(3|T|/(|T|-k))^2
 \cdot \log(3|T|/(|T|-k))$ instances  with single demand pairs.
 The former instance can be approximated within a factor of $O(k^2 \log
 (kd))$
 as above.
 Each of the latter instances admits a constant factor approximation
 using the algorithm presented in \cite{Nutov13activation}.
 These give $O(k^2 \log (kd))$-approximation for the original augmentation unless $k=|T|-o(|T|)$. 
 When $|T|=O(k)$ (including the case with $k=|T|-o(|T|)$), 
 the augmentation problem can be decomposed into $O(k^2)$ instances
 with single demand pairs, resulting in an $O(k^2)$-approximation for
 the augmentation problem.
 Recall that we pay factor $k$ for reducing PCNAP to the
 prize-collecting augmentation problem.
 Therefore, we have an $O(k^3 \log (kd))$-approximation algorithm for the
 subset node-connectivity PCNAP.
\end{proof}

Note that $\log (kd)=O(\log|V|)$ in Corollary~\ref{cor.node}.

\section{Conclusion}\label{sec.conclusion}

We have presented approximation algorithms for PCNAP.
Our algorithms are built on new formulations of LP relaxations, the primal-dual algorithm for computing spiders, and the potential
function for analyzing the greedy spider cover algorithm.

Our algorithms must solve the LP relaxation in order to decide which demand
pairs should be satisfied by solutions. In contrast, several primal-dual algorithms such as
those in \cite{BateniHL13,Konemann13CoRR} can manage this without solving LP
by generic LP solvers.
In other words, these algorithms are combinatorial.
We believe that it is challenging to design combinatorial algorithms for PCNAP.

\section*{Acknowledgements}

This work was supported by
JSPS KAKENHI Grant Number 25730008 in part.
The author thanks Zeev Nutov for sharing information on his paper \cite{Nutov12uncrossable}.

\end{document}